\documentclass[conf]{new-aiaa}

\usepackage[utf8]{inputenc}
\usepackage{textcomp}

\usepackage{tikz}
\usetikzlibrary{arrows.meta, positioning,calc,fit,backgrounds}
\usepackage{graphicx}
\usepackage{amsmath}
\usepackage{mathtools}
\usepackage{hyperref}
\usepackage[english]{babel}
\usepackage[version=4]{mhchem}
 \usepackage{amsmath}
\usepackage{siunitx}
\usepackage{blindtext}
\usepackage{hyperref}
\usepackage{longtable,tabularx}
 \usepackage{amsthm}
\usepackage[utf8]{inputenc}
\usepackage{graphicx}
\usepackage{subcaption}
\usepackage{amsmath}
\usepackage[version=4]{mhchem}
\usepackage{siunitx}
\usepackage{longtable,tabularx}
\usepackage{algorithm}
\usepackage{algpseudocode}
\usepackage{amsmath}
\usepackage{float}
\usepackage{titlesec}

\usepackage{xcolor}
\usepackage{textcomp}
\usepackage{graphicx}
\usepackage{amsmath}
\usepackage[version=4]{mhchem}
\usepackage{siunitx}
\usepackage{longtable,tabularx}
\usepackage{graphicx}
\usepackage{subcaption}
\usepackage{amsmath}
\usepackage{mathtools}
\usepackage[version=4]{mhchem}
\usepackage{siunitx}
\usepackage{longtable,tabularx}
\usepackage{algorithm}
\usepackage{algpseudocode}
\usepackage{amsmath}
\usepackage{amsfonts}
\usepackage{comment}
\usepackage{titlesec}
\usepackage{booktabs,subcaption,amsfonts,dcolumn}

\newcommand{\qhbddot}{\dot{\hat{q}}_{{\mathrm{B}/\mathrm{D}}}}
\newcommand{\bx}{{\boldsymbol{x}}}
\newcommand{\bu}{{\boldsymbol{u}}}

\newcommand{\qhdi}{{{\hat{q}}_{{\mathrm{D}/\mathrm{I}}}}}
\newcommand{\ohdi}{{\hat{\omega}_{{\mathrm{D}/\mathrm{I}}}}}
\newcommand{\qhbd}{{{\hat{q}}_{{\mathrm{B}/\mathrm{D}}}}}
\newcommand{\qhstarsqmones}{{{\hat{q}}_{{\mathrm{B}/\mathrm{D}}}}^{\star}\left({{\hat{q}}_{{\mathrm{B}/\mathrm{D}}}}^{\mathrm{s}}-\mathbf{1}^{\mathrm{s}}\right)}
\newcommand{\qhbds}{{\hat{q}}^\mathrm{s}_{{\mathrm{B}/\mathrm{D}}}}
\newcommand{\qhbdmone}{({{\hat{q}}_{{\mathrm{B}/\mathrm{D}}}}-\mathbf{1})}
\newcommand{\dqhbd}{{{\dot{\hat{q}}}_{{\mathrm{B}/\mathrm{D}}}}}
\newcommand{\ohbd}{{{\hat{\omega}}^\mathrm{B}_{{\mathrm{B}/\mathrm{D}}}}}
\newcommand{\dohbd}{{{\dot{\hat{\omega}}}_{{\mathrm{B}/\mathrm{D}}}}}
\newcommand{\ohbi}{{\hat{\omega}_{\mathrm{B/I}}}}
\newcommand{\ohdib}{\hat{\omega}_{\mathrm{D} / \mathrm{I}}^{\mathrm{B}}}
\newcommand{\ohbdb}{{\hat{\omega}^{\mathrm{B}}_{\mathrm{B/D}}}}
\newcommand{\ohbdi}{{\hat{\omega}^{\mathrm{B}}_{\mathrm{D/I}}}}

\newcommand{\qe}{\|{{\hat{q}}_{{\mathrm{B}/\mathrm{D}}}}-\mathbf{1}\|^2}

\newcommand{\qenorm}{\|\qhbd-\mathbf{1}\|^2}

\newcommand{\omegaenorm}{\|\ohbdb\|^2}
\newcommand{\omegaenormsqrt}{\|\ohbdb\|}
\newcommand{\omegae}{{{\hat{\omega}}^\mathrm{Bs}_{{\mathrm{B}/\mathrm{D}}}}}
\newtheorem{remark}{\textnormal{\textbf{Remark}}}
\newtheorem{theorem}{\bf Theorem}

\newtheorem{definition}{\bf Definition}

\usepackage{dirtytalk}

\newtheorem{lemma}{\bf Lemma}

\newcommand{\obdb}{{{\omega}^{\text{B}}_{\text{B/D}}}}
\newcommand{\qbd}{{{{q}}_{{\mathrm{B}/\mathrm{D}}}}}

\newcommand{\ohbdbdot}{{\dot{\hat{\omega}}^{\text{B}}_{\text{B/D}}}}
\newcommand{\dqbd}{{{\dot{{q}}}_{{\mathrm{B}/\mathrm{D}}}}}
\newcommand{\obdbdot}{{\dot{{\omega}}^{\text{B}}_{\text{B/D}}}}
\title{Semi-global Exponential Stability for Dual Quaternion Based Rigid-Body Tracking Control }

\author{Vrushabh Zinage\footnote{Graduate Student, Department of Aerospace Engineering and Engineering Mechanics, Student Member AIAA}, S P Arjun Ram\footnote{Graduate Research Assistant and Ph.D. candidate, Aerospace Engineering and Engineering Mechanics, The University of Texas at Austin, Austin. Email: arjun.ram@utexas.edu}, Maruthi R. Akella \footnote{Professor and Cockrell Family Chair in Engineering \#19, Aerospace Engineering and Engineering Mechanics, The University of Texas at Austin, Austin, Fellow AIAA. Email: makella@mail.utexas.edu} and Efstathios Bakolas \footnote{Associate Professor, Department of Aerospace Engineering and Engineering Mechanics, Senior Member AIAA. Email: bakolas@austin.utexas.edu}}
\affil{The University of Texas at Austin, Austin, Texas 78712}

\begin{document}

\maketitle

\begin{abstract}
Semi-Global Exponential Stability (SGES) is proved for the combined attitude and position rigid body motion tracking problem, which was previously only known to be asymptotically stable. Dual quaternions are used to jointly represent the rotational and translation tracking error dynamics of the rigid body. A novel nonlinear feedback tracking controller is proposed and a Lyapunov based analysis is provided to prove the semi-global exponential stability of the closed-loop dynamics. Our analysis does not place any restrictions on the reference trajectory or the feedback gains. This stronger SGES result aids in further analyzing the robustness of the rigid body system by establishing Input-to-State Stability (ISS) in the presence of time-varying additive and bounded external disturbances. Motivated by the fact that in many aerospace applications, stringent adherence to safety constraints such as approach path and input constraints is critical for overall mission success, we present a framework for safe control of spacecraft that combines the proposed feedback controller with Control Barrier Functions. Numerical simulations are provided to verify the SGES and ISS results and also showcase the efficacy of the proposed nonlinear feedback controller in several non-trivial scenarios including the Mars Cube One (MarCO) mission, Apollo transposition and docking problem, Starship flip maneuver, collision avoidance of spherical robots, and the rendezvous of SpaceX Dragon 2 with the International Space Station.
\end{abstract}

\section{Nomenclature}

{\renewcommand\arraystretch{1.0}
\noindent\begin{longtable*}{@{}l @{\quad=\quad} l@{}}
$0_{m\times n}$ & $m\times n$ zero matrix\\
$I_n$ & Identity matrix of size $n$\\
$\mathbb{Q},\mathbb{D}$  & Set of quaternions and dual quaternions respectively \\
$a,\;\hat{a}$ &    Quaternion and dual quaternion respectively \\
$1,\;0$ &    Quaternions $(0_{3\times 1},1)$ and $(0_{3\times 1},0)$ respectively \\
$q_{\mathrm{Y} /\mathrm{Z}}$ & Unit quaternion from the Z-frame to the Y-frame\\
$  \bar{\omega}_{\mathrm{Y} / Z}^X$ & Angular velocity of the Y-frame with respect to the Z-frame expressed in the X-frame  \\
$  \bar{v}_{\mathrm{Y} / \mathrm{Z}}^{\mathrm{X}}$ & Linear velocity of the origin of the Y-frame  with respect to the Z-frame expressed in the X-frame  \\
$\bar{r}_{\mathrm{Y} / \mathrm{Z}}^{\mathrm{X}}$ & Translation vector from the origin of the  Z-frame to the origin of the Y-frame expressed in the X-frame  \\
${\tau}^{\mathrm{B}}$ & Total external moment vector applied to the body about its center of mass expressed in the body frame \\
$f^{\mathrm{B}}$ & Total external force vector applied to the body expressed in the body frame \\
 $\hat{f}^\mathrm{B}$ & Dual input force $(f^{\mathrm{B}},0)+\epsilon({\tau}^{\mathrm{B}},0) $\\
$\mathcal{B}_R$& $\{x: \|x\|\leq R\}$ \\
 $J,\;m$ & Dual inertia matrix and mass of rigid body respectively\\
$J_M,\;J_m$ & Maximum and minimum eigenvalues of matrix $J$ \\
 $\mathbf{1},\;\mathbf{0}$ & Dual quaternions $1+\epsilon 0$ and $0+\epsilon 0$ respectively \\
 $\hat{q}_{\mathrm{Y}/\mathrm{Z}}$ & Unit dual quaternion from the Z-frame to Y-frame\\
 
\end{longtable*}}
\section{Introduction}

Any space program's maturity has historically hinged on their ability to perform rendezvous and docking (RVD) operations \cite{woffinden2007navigating_space_programs_1} which have gained significant interest due to applications like on-orbit refueling and assembly \cite{shen2005peer_1,chen2016output_2}. Pioneering U.S. and Soviet missions, like Gemini and Soyuz, aimed to demonstrate these essential capabilities, which later facilitated constructing and servicing space stations in low Earth orbit \cite{goodman2006history_space_programs_2}. Although the core technology has largely remained unchanged, it may not suffice for future missions \cite{woffinden2007navigating_space_programs_1}. The development of fully autonomous rendezvous systems for various domains, including low Earth orbit and low lunar orbit \cite{d2007orion_space_programs_3}, is imperative to advance space exploration and address upcoming challenges. RVD missions \cite{woffinden2007navigating_space_programs_1,shen2005peer_1,chen2016output_2,goodman2006history_space_programs_2} typically involve several main phases, with the final and proximity operations being crucial for mission complexity and safety. Various trajectory designs, control laws, and optimal planning schemes for these operations have been proposed in literature \cite{breger2008safe_4,singla2006adaptive_5}. Unlike traditional spacecraft control problems, RVD missions often require six-degree-of-freedom (6-DOF) control for the chaser spacecraft to synchronously track during critical phases.
Within this context, the MarCO (Mars Cube One) mission by NASA's Jet Propulsion Laboratory (JPL) \cite{schoolcraft2017marco} has marked a significant advancement. Launched in 2018, MarCO sent two CubeSats, small modular satellites, to Mars. Their role was to relay real-time communication during the InSight Mars lander's entry, descent, and landing (EDL) operations. This marked the first use of CubeSats in deep space, demonstrating their potential for future interplanetary missions. In recent years, the term "proximity operations" has become prevalent to describe a variety of space missions that necessitate one spacecraft maintaining close proximity to another space object. These missions may involve tasks such as inspection, health monitoring, surveillance, servicing, and refueling of a space asset by a different spacecraft. A significant challenge in autonomous proximity operations, whether cooperative or uncooperative, is the ability to independently and precisely track the changing relative positions and attitudes, or pose references, in relation to a moving target. This is crucial in order to prevent on-orbit collisions and accomplish the overall objectives of the mission.

A common thread underlying these advancements, from rendezvous and docking operations to the deployment of CubeSats in deep space, is the need for 6-DOF control that involves the precise joint control of both rotational and translational motion. Even though there is rich literature on providing various stability guarantees independently for attitude control and for translational control, the problem of jointly controlling both the position and attitude has only recently gained attention. Spacecraft rendezvous \cite{zinage2022minimum,dong2019adaptive,malyuta2020fast_rendezous}, spacecraft pose estimation \cite{filipe2015extended_pose_estimation,kaki2023real_pose_estimation}, pose tracking \cite{filipe2016pose_tracking}, powered descent guidance \cite{lee2015optimal_powered_desent,malyuta2019discretization_powered_descent} and spacecraft formation flying \cite{huang2017dual_spacecraft_formation_flying} are some important applications that require joint control of both rotational and translational motion. Dual quaternions \cite{brodsky1999dual_old} provide a compact way of representing not only the attitude but also the position of the body, making it a great choice for applications ranging from inertial navigation, formation flying, and robot kinematic chains to computer vision and animation. The dynamics of the tracking control problem for many of the aforementioned aero-mechanical systems tend to be nonlinear, making the controller design a challenging task. Precise characterization of the stability characteristics of the controller significantly enhances the chances of its successful adoption for any particular application. Dual quaternions \cite{brodsky1999dual_old,tsiotras2020dual_applications} have the added advantage of allowing a combined position and attitude control law to be written in a compact single form. Further, the governing rigid body dynamics and the designed feedback control laws based on combined position and attitude (characterized by a single dual quaternion) take into account the natural coupling between the translation and rotational motion of the rigid body.

Tsiotras~\cite{tsiotras1996stabilization} provides asymptotic stability guarantees for attitude stabilization using quaternions and Modified Rodriges Parameters (MRPs) without angular feedback. This asymptotic stability result was extended to the attitude tracking problem in \cite{akella2001rigid}. Recently, a stronger Uniform Exponential Stability (UES) result for attitude stabilization was provided using a novel Lyapunov analysis design, taking advantage of the structure of the MRP error dynamics, where the feedback controller is the MRPs based classical Proportional Derivative with feed-forward terms (PD+ controller) \cite{arjun2020uniform}. When using quaternions for attitude control, it is difficult to guarantee exponential convergence of the vector part of the quaternions to zero without any constraints or lower limits on the feedback gains, as is the case in \cite{WenDelgado} because this high gain feedback is necessary to dominate the nonlinear terms that show up in the error dynamics when using quaternions. In  \cite{arjun2020uniform}, the UES result based on MRPs was extended to quaternions. However, instead of using a direct Lyapunov analysis based on the quaternion kinematics, the argument was made that the vector part of the quaternion error is upper bounded by the MRP error, and hence exponentially converge when the MRPs do so. In \cite{dong2019adaptive}, a set of artificial potential functions are constructed and leveraged to incorporate information regarding spacecraft motion constraints such as field-of-view and approach path constraints that must be followed during the proximity operations. It is shown that by combining these functions with the designed feedback controller, the leader can be steered to the follower asymptotically while respecting its motion constraints. The authors of \cite{filipe2013simultaneous} propose a dual quaternion based nonlinear feedback controller to jointly stabilize both the position and attitude of the rigid body. Further, they provide asymptotic stability guarantees for rigid body dynamics. This result was extended to the tracking control problem for a rigid body in \cite{filipe2013rigid_tracking}.


In this work, we propose a novel feedback tracking control law and present a stronger Semi-Global Exponential Stability (SGES) guarantee for the dual quaternion-based rigid body tracking control problem where
the proposed control law actually guarantees semi-global exponential convergence of both attitude and position states of the rigid body to the desired trajectories and not just asymptotic stability as known before. 
Moreover, no extra conditions are placed on the controller gains, implying that the controller can be designed without any prior knowledge of the bounds on the reference trajectory or body inertia. This stronger semi-global exponential convergence property of the closed loop system is a previously unknown result and therefore presents itself as the major contribution of this paper. This result reaffirms the effectiveness of the controller and provides a compelling case for its use in the 6-DOF motion control of spacecraft, air vehicles, and robotic systems. The stronger exponential stability result also allows us to provide robustness guarantees of local Input-to-State Stability to the controller in the presence of unknown additive-bounded external disturbances. Such disturbances might result from inertia uncertainties, actuators, or measurement sensors, demonstrating the reliability of the proposed model for space applications. 

The rest of the paper is organized as follows. Section \ref{sec:prelim} provides  a brief overview of dual quaternion algebra and the tracking error dynamics. Section \ref{sec:control_design_and_stability_analysis} presents the control design and asymptotic stability analysis for the dual quaternion based tracking dynamics of the rigid body followed by our main result of semi-global exponential stability in Section \ref{sec:exp_stability}. Section \ref{sec:robustness_analysis} presents the robustness properties of the proposed nonlinear feedback controller in terms of Input-to-State Stability. Section \ref{sec:safety_and_motion_constraints} presents the methodology for safe control synthesis in the presence of motion and safety constraints. Section \ref{sec:results} presents the results on five realistic scenarios where we verify the SGES and ISS claims followed by concluding remarks in Section \ref{sec:conclusions}. 

\section{Preliminaries\label{sec:prelim}}
\subsection{Quaternion Algebra}
A quaternion $q$ can be represented by a pair $(\Bar{q},q_4)$ where $\Bar{q}\in\mathbb{R}^3$ is known as its vector part and $q_4\in\mathbb{R}$ as its scalar part. The set of quaternions is denoted by $\mathbb{Q}$. The set $\mathbb{Q}^v = \{q \in\mathbb{Q}: q_4 = 0\}$ denotes the set of quaternions with scalar part zero and the set $\mathbb{Q}^s = \{q \in\mathbb{Q}: \bar{q} = 0\}$ denotes the set of quaternions with zero vector part. Some of the basic operations on quaternions $a,b\in\mathbb{Q}$ are given below
\begin{align}
&\text {Addition: } a+b=\left(\Bar{a}+\Bar{b}, a_{4}+b_{4}\right),\nonumber\\
&\text {Multiplication by a scalar: } \lambda a=\left(\lambda \Bar{a}, \lambda a_{4}\right),\nonumber\\
&\text {Multiplication: } a b=\left(a_{4} \Bar{b}+b_{4} \Bar{a}+\Bar{a} \times \Bar{b}, a_{4} b_{4}-\Bar{a} \cdot \Bar{b}\right),\nonumber \\
&\text {Conjugation: } a^{\star}=\left(-\Bar{a}, a_{4}\right),\nonumber \\
&\text {Dot product: } a \cdot b=\frac{1}{2}\left(a^{\star} b+b^{\star} a\right)=\frac{1}{2}\left(a b^{\star}+b a^{\star}\right)=\left(\Bar{0}, a_{4} b_{4}+\Bar{a} \cdot \Bar{b}\right),\nonumber \\
&\text {Cross product: } a \times b=\frac{1}{2}\left(a b-b^{\star} a^{\star}\right)=\left(b_{4} \Bar{a}+a_{4} \Bar{b}+\Bar{a} \times \Bar{b}, 0\right), \nonumber\\
&\text {Norm: }\|a\|^{2}=a a^{\star}=a^{\star} a=a \cdot a=\left(\Bar{{0}}, a_{4}^{2}+\Bar{a} \cdot \Bar{a}\right),\nonumber \\
&\text { Scalar part: } \operatorname{sc}(a)=\left(\Bar{0}, a_{4}\right) \in \mathbb{Q}^{s}\nonumber \\
&\text { Vector part: } \operatorname{vec}(a)=(\Bar{a}, 0) \in \mathbb{Q}^{v}\nonumber\\
&\text{Multiplication by a matrix: } M\star q=(M_{11} \bar{q} + M_{12}q_4, M_{21} \bar{q} + M_{22}q_4)\nonumber
\end{align}
where $M= \Big[\begin{smallmatrix} M_{11} & M_{12} \\
M_{21} & M_{22} \end{smallmatrix} \Big]\in\mathbb{R}^{4\times 4}$, where $M_{11}\in \mathbb{R}^{3 \times 3},\; M_{21}\in \mathbb{R}^{1 \times 3}, M_{12} \in \mathbb{R}^{3 \times 1}, \; M_{22} \in \mathbb{R}$.
\subsection{Dual Quaternion Algebra}
A dual quaternion $\hat{q}$ can be represented by a pair $({q}_r,{q}_d)$ where ${q}_r\in\mathbb{Q}$ and ${q}_d\in\mathbb{Q}$ are the real and dual parts of $\hat{q}$, respectively. Let the set of dual quaternions be denoted by $\mathbb{D}$ and let also $\mathbb{D}^v_d = \{\hat{q} : \hat{q} = q_r + \epsilon q_d,\; q_r ,\; q_d \in \mathbb{Q}^v \}$ denote the set of dual quaternions where the scalar component of the real and dual parts of $\hat{q}$ are zero. The set $\mathbb{D}_d^s$ denotes the set of dual quaternions where the vector component of the real and dual parts of $\hat{q}$ are zero. A dual quaternion $\hat{q}$ can also be represented as $\hat{q}=q_r+\epsilon q_d$ where $\epsilon^2=0$ and $\epsilon\neq0$. A list of basic operations on dual quaternions $\hat{a},\;\hat{b}\in\mathbb{D}$ are given below \cite{filipe2013rigid_tracking}:
\begin{align}
&\text {Addition: } \hat{a}+\hat{b}=\left(a_{r}+b_{r}\right)+\epsilon\left(a_{d}+b_{d}\right),\nonumber\\
&\text {Multiplication by a scalar: } \lambda \hat{a}=\left(\lambda a_{r}\right)+\epsilon\left(\lambda a_{d}\right),\nonumber\\
&\text {Multiplication: } \hat{a} \hat{b}=\left(a_{r} b_{r}\right)+\epsilon\left(a_{r} b_{d}+a_{d} b_{r}\right),\nonumber \\
&\text {Conjugation: } \hat{a}^{\star}=a_{r}^{\star}+\epsilon a_{d}^{\star},\nonumber \\
 &\text {Swap: } \hat{a}^{\mathrm{s}}=a_{d}+\epsilon a_{r},\nonumber \\
&\text {Dot product: } \hat{a} \cdot \hat{b}=\frac{1}{2}\left(\hat{a}^{\star} \hat{b}+\hat{b}^{\star} \hat{a}\right)=\frac{1}{2}\left(\hat{a} \hat{b}^{\star}+\hat{b} \hat{a}^{\star}\right) =a_{r} \cdot b_{r}+\epsilon\left(a_{d} \cdot b_{r}+a_{r} \cdot b_{d}\right),\nonumber\\
&\text {Cross product: } \hat{a} \times \hat{b}=\frac{1}{2}\left(\hat{a} \hat{b}-\hat{b}^{\star} \hat{a}^{\star}\right)=a_{r} \times b_{r}+\epsilon\left(a_{d} \times b_{r}+a_{r} \times b_{d}\right),\nonumber \\
&\text {Circle product: }\hat{a}\circ\hat{b}=a_r\cdot b_r+a_d\cdot b_d,  \nonumber\\
&\text {Norm: }\|\hat{a}\|^{2}=a_r\cdot a_r+a_d\cdot a_d , \nonumber\\
&\text {Dual norm: }\|\hat{a}\|_{d}^{2}=\hat{a} \hat{a}^{\star}=\hat{a}^{\star} \hat{a}=\hat{a} \cdot \hat{a} \nonumber\\
&\quad\quad\quad\quad\quad\quad\quad\;\;=\left(a_{r} \cdot a_{r}\right)+\epsilon\left(2 a_{r} \cdot a_{d}\right),  \nonumber\\
&\text { Scalar part: } \operatorname{sc}(\hat{a})=\operatorname{sc}(a_r)+\epsilon\operatorname{sc}(a_d)\in\mathbb{D}_d^s\nonumber \\
&\text { Vector part: } \operatorname{vec}(\hat{a})=\operatorname{vec}(a_r)+\epsilon\operatorname{vec}(a_d)\mathbb{D}_d^v\nonumber\\
& \text{Multiplication by matrix}: M \star \hat{q}=\left(M_{11} \star q_{r}+M_{12} \star q_{d}\right)+\epsilon\left(M_{21} \star q_{r}+M_{22} \star q_{d}\right),\nonumber
\end{align}
where $M= \Big[\begin{smallmatrix} M_{11} & M_{12} \\
M_{21} & M_{22} \end{smallmatrix} \Big]\in\mathbb{R}^{8\times 8}$, where $M_{11}, M_{12}, M_{21}, M_{22} \in \mathbb{R}^{4 \times 4}$.

Some other properties that follow from the dual quaternion algebra are as follows:
\begin{subequations}
\begin{equation}
\hat{a} \circ(\hat{b} \hat{c})=\hat{b}^{\mathrm{s}} \circ\left(\hat{a}^{\mathrm{s}} \hat{c}^*\right)=\hat{c}^{\mathrm{s}} \circ\left(\hat{b}^* \hat{a}^{\mathrm{s}}\right) \in \mathbb{R}, \quad \hat{a}, \hat{b}, \hat{c} \in \mathbb{D}_d,
\label{eqn:circ_prod_property}
\end{equation}
\begin{equation}
\hat{a} \circ(\hat{b} \times \hat{c})=\hat{b}^{\mathrm{s}} \circ\left(\hat{c} \times \hat{a}^{\mathrm{s}}\right)=\hat{c}^{\mathrm{s}} \circ\left(\hat{a}^{\mathrm{s}} \times \hat{b}\right), \quad \hat{a}, \hat{b}, \hat{c} \in \mathbb{D}_d^v,
\label{eqn:property_4}
\end{equation}
\begin{equation}
(M \star \hat{a}) \circ \hat{b}=\hat{a} \circ\left(M^{\top} \star \hat{b}\right),\quad \hat{a}, \hat{b} \in \mathbb{D}_d,\quad M \in \mathbb{R}^{8 \times 8}
\label{eqn:property_9_in_other_paper}
\end{equation}
\begin{equation}
\hat{a} \times \hat{b} =-\hat{b} \times \hat{a}, \quad \hat{a}, \hat{b} \in \mathbb{D}_d .
\label{eqn:property_5_in_other_paper}
\end{equation}
\end{subequations}
\subsection{Tracking error dynamics}
The error dynamics for the rigid body motion expressed in terms of dual quaternions is given by \cite{filipe2013rigid_tracking}
\begin{align}
&{\dqhbd}=\frac{1}{2}\qhbd\ohbdb\nonumber\\
&\left(\dot{\hat{\omega}}_{\mathrm{B} / \mathrm{D}}^{\mathrm{B}}\right)^{\mathrm{s}}=\left(J\right)^{-1} \star\left(\hat{f}^{\mathrm{B}}-\left(\ohbdb+\ohdib\right) \times\left(J\star\left(\left(\ohbdb\right)^{\mathrm{s}}+\left(\ohdib\right)^{\mathrm{s}}\right)\right)\right)\nonumber \\
&\quad\quad\left.-J\star\left(\qhbd^{\star} \dot{\hat{\omega}}_{\mathrm{D} / \mathrm{I}}^{\mathrm{B}} \qhbd\right)^{\mathrm{s}}-J\star\left(\ohdib \times \ohbdb\right)^{\mathrm{s}}\right)
\label{eqn:system_dynamics}
\end{align}
where $\qhbd={\hat{q}}_{{\mathrm{D}/\mathrm{I}}}^\star\hat{q}_{{\mathrm{B}/\mathrm{I}}},\;\hat{q}_{{\mathrm{B}/\mathrm{I}}}={q}_{{\mathrm{B}/\mathrm{I}}}+\epsilon\frac{1}{2}{q}_{{\mathrm{B}/\mathrm{I}}}{r}^\mathrm{B}_{{\mathrm{B}/\mathrm{I}}}$, ${r}^\mathrm{B}_{{\mathrm{B}/\mathrm{I}}}=(\Bar{r}^\mathrm{B}_{{\mathrm{B}/\mathrm{I}}},0)$ and $\ohbd=\hat{\omega}_{\text{B/I}}-\hat{\omega}_{\mathrm{D/I}}$. The rotation quaternion and the translation vector of the rigid body with respect to the inertial frame are represented by ${q}_{{\mathrm{B}/\mathrm{I}}}$ and $\Bar{r}^\mathrm{B}_{{\mathrm{B}/\mathrm{I}}}$ respectively. The desired dual quaternion $\hat{q}_{{\mathrm{D}/\mathrm{I}}}$ is defined similarly. Further, $\ohbi=\omega_\mathrm{B/I}+\epsilon v_\mathrm{B/I}$ where $\omega_{{\mathrm{B}/\mathrm{I}}}=(\Bar{\boldsymbol{\omega}},0)\in\mathbb{Q}^v$ and $v_{\mathrm{B}/\mathrm{I}}=(\Bar{\boldsymbol{v}},0)\in\mathbb{Q}^v$ where $\Bar{\boldsymbol{\omega}}\in\mathbb{R}^3$ and $\Bar{\boldsymbol{v}}\in\mathbb{R}^3$ are the angular and linear velocities of the rigid body with respect to the inertial frame respectively. The input dual force $\hat{{f}}^\mathrm{B}:=(f^{\mathrm{B}},0)+\epsilon({\tau}^{\mathrm{B}},0)$ is the total dual input applied to the rigid body where $f^\mathrm{B}$ and $\tau^\mathrm{B}$ represent the force and the torque applied respectively to the rigid body. $J$ is the dual inertia matrix defined as follows
\begin{align}
    J=\begin{bmatrix}
    mI_3&0_{3\times 1}&0_{3\times 3}&0_{3\times 1}\\
    0_{1\times 3}&1&0_{1\times 3}&0\\
    0_{3\times 3}&0_{3\times 1}&\Bar{I}^\mathrm{B}&0_{3\times 1}\\
    0_{1\times 3}&0&0_{1\times 3}&1
    \end{bmatrix}\nonumber
\end{align}
where $\Bar{I}^\mathrm{B}$ is the moment of inertia of the rigid body with respect to its center of mass, $m$ is the mass of the body and $I_3$ is the identity matrix.

\begin{remark}
\normalfont If the translation component $t_\mathrm{B}=t_\mathrm{D}=0$ and linear velocity $v_\mathrm{B}=v_\mathrm{D}=0$, then the dual part of $\ohbd$ and $\qhbd$ are equal to zero. In that case, the system dynamics \eqref{eqn:system_dynamics} simplifies to
\begin{align}
{\dqbd}&=\frac{1}{2}\qbd\obdb\nonumber\\
\left({\obdbdot}\right)^{\mathrm{s}}&=\left(N^{\mathrm{B}}\right)^{-1} \star\left({\tau}^{\mathrm{B}}-\left(\obdb+{\omega}_{\mathrm{D} / \mathrm{I}}^{\mathrm{B}}\right) \times\left(N^{\mathrm{B}} \star\left(\left({\omega}_{\mathrm{B} / \mathrm{D}}^{\mathrm{B}}\right)^{\mathrm{s}}+\left({\omega}_{\mathrm{D} / \mathrm{I}}^{\mathrm{B}}\right)^{\mathrm{s}}\right)\right)\right)\nonumber \\
&\left.-N^\mathrm{B} \star\left({q}_{\mathrm{B} / \mathrm{D}}^{\star} \dot{{\omega}}_{\mathrm{D} / \mathrm{I}}^{\mathrm{D}} {q}_{\mathrm{B} / \mathrm{D}}\right)^{\mathrm{s}}-N^{\mathrm{B}} \star\left({\omega}_{\mathrm{D} / \mathrm{I}}^{\mathrm{B}} \times {\omega}_{\mathrm{B} / \mathrm{D}}^{\mathrm{B}}\right)^{\mathrm{s}}\right)
\end{align}
where $N^{\mathrm{B}}$ is given by
\begin{align}
    N^{\mathrm{B}}=\left[\begin{array}{cc}
\Bar{I}^\mathrm{B} & 0_{3 \times 1} \\
0_{1 \times 3} & 1 
\end{array}\right],\nonumber
\end{align}
which can be seen to have four distinct terms in angular velocity dynamics: the torque, a cross product of the angular velocity and momentum, the rate of change of the desired velocity rotated to the body frame, and a cross product between the angular velocity error and the velocity. This is algebraically similar in structure to the error dynamics of MRPs as given in \cite{arjun2020uniform}.
\label{Remark1}
\end{remark}

\section{Control Design and Stability Analysis\label{sec:control_design_and_stability_analysis}}
In this section, we first propose the nonlinear feedback tracking controller and subsequently establish the asymptotic and semi-global exponential stability property of the closed-loop dynamics. Finally, we discuss the relation of the proposed nonlinear feedback control law with optimal control.
\subsection{Control Design\label{sec:feedback_control}}
The proposed nonlinear feedback tracking control law is given by
\begin{align}
\hat{f}^\mathrm{B} &=- k_{p}\frac{\qhstarsqmones}{1+\qenorm}-k_{d}\left(\ohbdb\right)^{\mathrm{s}}+J\star\left(\qhbd^{\star} \dot{\hat{\omega}}_{\mathrm{D} / \mathrm{I}}^{\mathrm{D}} \qhbd\right)^{\mathrm{s}} +\ohdib \times\left(J\star\left(\ohdib\right)^{\mathrm{s}}\right), \quad k_{p}, k_{d}>0
\label{eqn:feedback_control_law}
\end{align}
Note that the feedback control law proposed in \cite{filipe2013rigid_tracking} only guarantees asymptotic stability for the rigid body tracking error dynamics \eqref{eqn:system_dynamics}. Further, the control law in \cite{filipe2013rigid_tracking} is different from \eqref{eqn:feedback_control_law} and is given by 
\begin{align}
   \hat{f}^\mathrm{B} &=- k_{p}\operatorname{vec}\left({\qhstarsqmones}\right)-k_{d}\left(\ohbdb\right)^{\mathrm{s}}+J\star\left(\qhbd^{\star} \dot{\hat{\omega}}_{\mathrm{D} / \mathrm{I}}^{\mathrm{D}} \qhbd\right)^{\mathrm{s}} +\ohdib \times\left(J\star\left(\ohdib\right)^{\mathrm{s}}\right), \quad k_{p}, k_{d}>0 
\end{align}
In the subsequent sections, we show that the proposed feedback control law \eqref{eqn:feedback_control_law} guarantees asymptotic as well as semi-global exponential stability for the dual quaternion based closed-loop tracking error dynamics of the rigid body represented by \eqref{eqn:system_dynamics}. Furthermore, using the proposed Lyapunov function, we show Input to State Stability (ISS) for the rigid body dynamics in the presence of time-varying additive bounded disturbances.
\subsection{Asymptotic Stability\label{sec:asym_stability}}
As a first step to proving exponential stability, we  leverage Lyapunov analysis to initially prove asymptotic stability. To that end, consider the following candidate Lyapunov functions $V_0$:

\begin{align}
V_0=&k_{p}\text{ln}\left(1+\qhbdmone \circ\qhbdmone\right)+\frac{1}{2}\left(\ohbdb\right)^{\mathrm{s}} \circ\left(J\star\left(\ohbdb\right)^{\mathrm{s}}\right)
\label{eqn:V0}
\end{align}

\begin{theorem}
 \normalfont For $V_0$ defined in \eqref{eqn:V0}, $\dot{V}_0=-k_{d}\left(\ohbdb\right)^{\mathrm{s}} \circ$ $\left(\ohbdb\right)^{\mathrm{s}}=-k_d\|\left(\ohbdb\right)^{\mathrm{s}}\|^2 < 0$, for all $\left(\qhbd, \ohbdb\right) \in \mathbb{D}_{d}^{u}\footnote{$\mathbb{D}_{d}^{u}=\{\hat{q}\in\mathbb{D}_d:\;\hat{q}.\hat{q}=\hat{q}^\star \hat{q}=\hat{q}\hat{q}^\star=1\}$} \times \mathbb{D}_{d}^{v} \backslash\{\mathbf{1}, \mathbf{0}\}$\footnote{With slight abuse of notation we denote a dual quaternion $(\Bar{0},a)+\epsilon(\Bar{0},0)$ by $\boldsymbol{a}$.}.
 \label{thm:as}
\end{theorem}
\begin{proof}
The time derivative of $V_0$ along system trajectories \eqref{eqn:system_dynamics} is given by
\begin{align}
    \dot{V}_0=\frac{2k_p(\qhbd-\mathbf{1})\circ \qhbddot}{1+\qe}+\left(\ohbdb\right)^{\mathrm{s}} \circ\left(J\star\left(\ohbdbdot\right)^{\mathrm{s}}\right)
    \label{eqn:V0_dot_expression}
\end{align}
Substituting the feedback control law \eqref{eqn:feedback_control_law} and using \eqref{eqn:system_dynamics} and property \eqref{eqn:circ_prod_property} yields 
\begin{align}
\dot{V}_0=&\left(\ohbdb\right)^{\mathrm{s}} \circ\left(-k_{d}\left(\ohbdb\right)^{\mathrm{s}}\right)+\underbrace{\left(\ohbdb\right)^{\mathrm{s}} \circ{\left(k_{p} \frac{\qhbd^{\star}\left(\qhbd^{\mathrm{s}}-\mathbf{1}^{\mathrm{s}}\right)}{1+\qenorm}-k_{p}\frac{ \qhbd^{\star}\left(\qhbd^{\mathrm{s}}-\mathbf{1}^{\mathrm{s}}\right)}{1+\qenorm}\right)}}_{\text{Term 1}} \nonumber\\
&+\underbrace{\left(\ohbdb\right)^{\mathrm{s}} \circ\left(-\left(\ohbdb+\ohdib\right) \times\left(J\star\left(\left(\ohbdb\right)^{\mathrm{s}}+\left(\ohdib\right)^\mathrm{s}\right)\right)-J\star\left(\ohdib \times \ohbdb\right)^{\mathrm{s}}+\ohdib \times\left(J\star\left(\ohdib\right)^{\mathrm{s}}\right)\right)}_{\text{Term 2}}
\label{eqn:aymstability}
\end{align}
Note that term 1 of the RHS of \eqref{eqn:aymstability} is equal to zero because it is the circle product of a dual vector quaternion with a zero dual quaternion. Next, we show that term 2 is also equal to zero:

\begin{align}
&\left(\ohbdb\right)^{\mathrm{s}} \circ\left(-\left(\ohbdb+\ohdib\right) \times\left(J\star\left(\left(\ohbdb\right)^{\mathrm{s}}+\left(\ohdib\right)\right)^{s}\right)-J\star\left(\ohdib \times \ohbdb\right)^{\mathrm{s}}\right.\nonumber\\
&\left.+\ohdib \times\left(J\star\left(\ohdib\right)^{\mathrm{s}}\right)\right)\nonumber\\
&=\left(\left(\hat{\omega}_{\mathrm{B} / \mathrm{I}}^{\mathrm{B}}\right)^{\mathrm{s}}-\left(\ohdib\right)^{\mathrm{s}}\right) \circ\left(-\hat{\omega}_{\mathrm{B} / \mathrm{I}}^{\mathrm{B}} \times\left(J\star\left(\hat{\omega}_{\mathrm{B} / \mathrm{I}}^{\mathrm{B}}\right)^{\mathrm{s}}\right)-J\star\left(\ohdib \times\left(\hat{\omega}_{\mathrm{B} / \mathrm{I}}^{\mathrm{B}}-\ohdib\right)\right)^{\mathrm{s}}+\ohdib \times\left(J\star\left(\ohdib\right)^{\mathrm{s}}\right)\right) \nonumber\\
&=\left(\hat{\omega}_{\mathrm{B} / \mathrm{I}}^{\mathrm{B}}\right)^{\mathrm{s}} \circ\left(-\hat{\omega}_{\mathrm{B} / \mathrm{I}}^{\mathrm{B}} \times\left(J\star\left(\hat{\omega}_{\mathrm{B} / \mathrm{I}}^{\mathrm{B}}\right)^{\mathrm{s}}\right)-J\star\left(\ohdib \times \hat{\omega}_{\mathrm{B} / \mathrm{I}}^{\mathrm{B}}\right)^{\mathrm{s}}+\ohdib \times\left(J\star\left(\ohdib\right)^{\mathrm{s}}\right)\right)\nonumber \\
&=-\underbrace{\left(\hat{\omega}_{\mathrm{B} / \mathrm{I}}^{\mathrm{B}}\right)^{\mathrm{s}}\circ\left(\hat{\omega}_{\mathrm{B} / \mathrm{I}}^{\mathrm{B}} \times\left(J\star\left(\hat{\omega}_{\mathrm{B} / \mathrm{I}}^{\mathrm{B}}\right)^{\mathrm{s}}\right)\right)}_{\text{Term 3}}-\underbrace{\left(\hat{\omega}_{\mathrm{B} / \mathrm{I}}^{\mathrm{B}}\right)^{\mathrm{s}} \circ\left(J\star\left(\ohdib \times \hat{\omega}_{\mathrm{B} / \mathrm{I}}^{\mathrm{B}}\right)^{\mathrm{s}}\right)}_{\text{Term 4}}+\underbrace{\left(\hat{\omega}_{\mathrm{B} / \mathrm{I}}^{\mathrm{B}}\right)^{\mathrm{s}} \circ\left(\ohdib \times\left(J\star\left(\ohdib\right)^{\mathrm{s}}\right)\right)}_{\text{Term 5}}\nonumber \\
&+\underbrace{\left(\ohdib\right)^{\mathrm{s}} \circ\left(\hat{\omega}_{\mathrm{B} / \mathrm{I}}^{\mathrm{B}} \times\left(J\star\left(\hat{\omega}_{\mathrm{B} / \mathrm{I}}^{\mathrm{B}}\right)^{\mathrm{s}}\right)\right)}_{\text{Term 6}}+\underbrace{\left(\ohdib\right)^{\mathrm{s}} \circ\left(J\star\left(\ohdib \times \hat{\omega}_{\mathrm{B} / \mathrm{I}}^{\mathrm{B}}\right)^{\mathrm{s}}\right)}_{\text{Term 7}}-\underbrace{\left(\ohdib\right)^{\mathrm{s}} \circ\left(\ohdib \times\left(J\star\left(\ohdib\right)^{\mathrm{s}}\right)\right)}_{\text{Term 8}}\nonumber
\end{align}
Note that terms 3 and 8 are equal to zero as well in view of properties \eqref{eqn:circ_prod_property} and \eqref{eqn:property_4}. Now,
\begin{align}
\text{Term 4+Term 5+Term 6+Term 7}=&-\left(J\star\left(\hat{\omega}_{\mathrm{B} / \mathrm{I}}^{\mathrm{B}}\right)^{\mathrm{s}}\right)^{\mathrm{s}} \circ\left(\ohdib \times \hat{\omega}_{\mathrm{B} / \mathrm{I}}^{\mathrm{B}}\right)+\left(\hat{\omega}_{\mathrm{B} / \mathrm{I}}^{\mathrm{B}}\right)^{\mathrm{s}} \circ\left(\ohdib \times\left(J\star\left(\ohdib\right)^{\mathrm{s}}\right)\right)\nonumber \\
&+\left(\ohdib\right)^{\mathrm{s}} \circ\left(\hat{\omega}_{\mathrm{B} / \mathrm{I}}^{\mathrm{B}} \times\left(J\star\left(\hat{\omega}_{\mathrm{B} / \mathrm{I}}^{\mathrm{B}}\right)^{\mathrm{s}}\right)\right)+\left(J\star\left(\ohdib\right)^{\mathrm{s}}\right)^{\mathrm{s}} \circ\left(\ohdib \times \hat{\omega}_{\mathrm{B} / \mathrm{I}}^{\mathrm{B}}\right)
\label{eqn:four_and_seven}
\end{align}
Finally, using properties \eqref{eqn:property_4} and \eqref{eqn:property_5_in_other_paper}, we have
\begin{align}
\text{Term 4+Term 5+Term 6+Term 7}=&-\left(\ohdib\right)^{\mathrm{s}} \circ\left(\hat{\omega}_{\mathrm{B} / \mathrm{I}}^{\mathrm{B}} \times\left(J\star\left(\hat{\omega}_{\mathrm{B} / \mathrm{I}}^{\mathrm{B}}\right)^{\mathrm{s}}\right)\right)+\left(\hat{\omega}_{\mathrm{B} / \mathrm{I}}^{\mathrm{B}}\right)^{\mathrm{s}} \circ\left(\ohdib \times\left(J\star\left(\ohdib\right)^{\mathrm{s}}\right)\right)\nonumber \\
&+\left(\ohdib\right)^{\mathrm{s}} \circ\left(\hat{\omega}_{\mathrm{B} / \mathrm{I}}^{\mathrm{B}} \times\left(J\star\left(\hat{\omega}_{\mathrm{B} / \mathrm{I}}^{\mathrm{B}}\right)^{\mathrm{s}}\right)\right)-\left(\hat{\omega}_{\mathrm{B} / \mathrm{I}}^{\mathrm{B}}\right)^{\mathrm{s}} \circ\left(\ohdib \times\left(J\star\left(\ohdib\right)^{\mathrm{s}}\right)\right)=\mathbf{0}\nonumber
\end{align}

Therefore, the time derivative of the Lyapunov function along the closed-loop system that results from application of control input \eqref{eqn:feedback_control_law} in \eqref{eqn:system_dynamics} is equal to $\dot{V}_0=-k_{d}\left(\ohbdb\right)^{\mathrm{s}} \circ$ $\left(\ohbdb\right)^{\mathrm{s}}=-k_d\|\left(\ohbdb\right)^{\mathrm{s}}\|^2 \leq 0$, for all $\left(\qhbd, \ohbdb\right) \in \mathbb{D}_{d}^{u} \times \mathbb{D}_{d}^{v} \backslash\{\mathbf{1}, \mathbf{0}\}$. By integrating both sides of $\dot{V}=-k_d\left(\hat{\omega}_{\mathrm{BD}}^{\mathrm{B}}\right)^{\mathrm{s}} \circ\left(\hat{\omega}_{\mathrm{BD}}^{\mathrm{B}}\right)^{\mathrm{s}} \leq 0$, we obtain 
\begin{align}
\underset{t\rightarrow\infty}{\lim}\;\int_0^t k_d\left(\hat{\omega}_{\mathrm{BD}}^{\mathrm{B}}(\tau)\right)^{\mathrm{s}} \circ\left(\hat{\omega}_{\mathrm{BD}}^{\mathrm{B}}(\tau)\right)^{\mathrm{s}} \mathrm{d} \tau \leq V(0) .
\label{eqn:v_int}
\end{align}
Since $\hat{q}_{\mathrm{B/D}}, \hat{\omega}_{\mathrm{B/D}}^{\mathrm{B}}\;\text{and}\; \hat{\omega}_{\mathrm{D} /\mathrm{I}}^{\mathrm{B}}\in \mathcal{L}_{\infty}$\footnote{The $\mathcal{L}_{\infty}$-norm of a function $\hat{f}:[0, \infty) \rightarrow$ $\mathbb{D}_d$ is defined as $\|\hat{f}\|_{\infty}=\underset{t\geq 0}{\sup}\;\|\hat{f}(t)\|$. The dual quaternion $\hat{f} \in \mathcal{L}_{\infty}$, if and only if $\|\hat{f}\|_{\infty}<\infty$.}, from \eqref{eqn:feedback_control_law} it follows that $\hat{f}^{\mathrm{B}} \in \mathcal{L}_{\infty}$ as well. From \eqref{eqn:system_dynamics}, it also follows that $\dot{\hat{\omega}}_{\mathrm{B}, \mathrm{D}}^{\mathrm{B}} \in \mathcal{L}_{\infty}$. Along with \eqref{eqn:v_int}, this yields $\hat{\omega}_{\mathrm{B} / \mathrm{D}}^{\mathrm{B}}(t) \rightarrow 0$ as $t \rightarrow \infty$. 
\end{proof}
\subsection{Semi-global Exponential Stability Result\label{sec:exp_stability}}
In this section, we show that the feedback control law \eqref{eqn:feedback_control_law} renders the tracking error dynamics \eqref{eqn:system_dynamics} semi-global exponential stability stable. To that end, we first state the definition of semi-global exponential stability.
\begin{definition}
\normalfont [Definition 5.10, \cite{sastry2013nonlinear}] (\textbf{Semi-global exponential stability}) Let $\bx=0$ be an equilibrium point for a nonlinear system $\dot{\bx}=f(\bx)$ and $\mathcal{B}_R$ be a ball of radius $R$ centered at the origin. Then, $\bx=0$ is semi-globally exponentially stable equilibrium point, if for a given $R>0$, there exist $m(R)>0$ and $\alpha(R)>0$ such that
\begin{align}
    \| \bx(t) \|\leq m(R)e^{-\alpha(R)(t-t_0)}\|x(t_0)\|,
\end{align}
for all $\bx(t_0)\in\mathcal{B}_R$ and $t\geq t_0\geq 0$.
\label{thm:es}
\end{definition}
For our analysis, for $\mathcal{B}_R$, we will consider the following set
\begin{align}
\mathcal{B}_R=\left\{\bx=(\qhbd-\mathbf{1},\ohbd):\;\qenorm+\|\ohbd\|^2\leq R^2\right\}
\label{eqn:ball_defn}
\end{align}
where $R>0$.
Note that $(\qhbd,\ohbd)=(\mathbf{1},\mathbf{0})$ is the equilibrium point for \eqref{eqn:system_dynamics}. Therefore, we consider $\bx(t):=(\qhbd-\mathbf{1},\ohbd)=(0,0)$ to be the equilibrium point. To streamline this presentation, we drop the time indexing for the states $\qhbd-\mathbf{1}$ and $\ohbdb$. Without loss of generality, we assume that $t_0=0$. Unless otherwise mentioned $\bx(t)$ indicates the states $(\qhbd-\mathbf{1},\ohbd)$ for any time instant $t\geq 0$. 
{ Note that in practical applications, the distinction between semi-global exponential stability and global exponential stability becomes negligible. Semi-global exponential stability guarantees that for any given bounded set of initial conditions i.e. within a ball of arbitrary radius $R>0$, the proposed feedback control law ensures the system's exponential stability. This is sufficient in practice, where initial conditions are naturally bounded. Global exponential stability, on the other hand, extends this guarantee to all initial conditions $\bx_0\in\mathbb{R}^n$. }

The following lemma will be used in the proof of semi-global exponential stability.
\begin{lemma}
\normalfont $\|\qhstarsqmones\|^2\geq \frac{\qe}{2}$
\label{lemma:inequality_new}
\end{lemma}
\begin{proof}
    Let $\qhbd:=q_r+\epsilon q_d$ where $q_d=\frac{1}{2}q_rt^B$ and $t^B=(\Bar{t}^B,0)$. We have 
    \begin{align}
        \qhbd^\star(\qhbd^s-\mathbf{1}^\mathrm{s})=q_r^\star q_d+\epsilon[q_r^\star(q_r-1)+q_d^\star q_d]
        \label{eqn:control_first_term}
    \end{align}
    Using Lemma 6 from \cite{salgueiro2014nonlinear_thesis} and the fact that $q_r$ is a unit  quaternion, we have
    \begin{align}
        &q_r^\star q_d=q_r^\star(\frac{1}{2}q_rt^B)=\frac{1}{2}q_r^\star q_rt^B=\frac{1}{2}t^B,\quad q_r^\star q_r=1,\quad q_d^\star q_d=\left(\frac{1}{2}q_rt^B\right)^\star\left(\frac{1}{2}q_rt^B\right)=\frac{1}{4}\|\Bar{t}^B\|^2\label{eqn:seperation_terms}
    \end{align}
    Therefore, substituting \eqref{eqn:seperation_terms} in \eqref{eqn:control_first_term}, we have
    \begin{align}
        \qhbd^\star(\qhbd^s-\mathbf{1}^\mathrm{s})=\frac{1}{2}t^B+\epsilon\left(1+\frac{1}{4}\|\Bar{t}^B\|^2-q_r^\star\right)\label{eqn:q_star_expression}
    \end{align}
 Consequently, from \eqref{eqn:q_star_expression}
    \begin{align}
        \|\qhstarsqmones\|^2=&\frac{3}{4}\|\Bar{t}^B\|^2+\|\Bar{q}_r\|^2+q^2_{r4}+1+\frac{1}{16}\|\Bar{t}^B\|^4-2q_{r4}-\frac{1}{2}q_{r4}\|\Bar{t}^B\|^2\nonumber\\
        &=\frac{3}{4}\|\Bar{t}^B\|^2+2(1-q_{r4})+\frac{1}{16}\|\Bar{t}^B\|^4-\frac{1}{2}q_{r4}\|\Bar{t}^B\|^2\nonumber\\
        &=\frac{1}{4}\|\Bar{t}^B\|^2+2(1-q_{r4})+\frac{1}{16}\|\Bar{t}^B\|^4+\frac{1}{2}\left(1-q_{r4}\right)\|\Bar{t}^B\|^2
        \label{eqn:lemma1_first}
    \end{align}
   In addition,
    \begin{align}
        \frac{\qe}{2}=(1-q_{r4})+\frac{1}{8}q_{r4}^2\|\Bar{t}^B\|^2+\frac{\|\Bar{q}_r\times \Bar{t}^B\|^2}{8}+\frac{\|\Bar{q}_r\cdot \Bar{t}^B\|^2}{8}
        \label{eqn:lemma1_second}
    \end{align}
    Now, combining \eqref{eqn:lemma1_first} and \eqref{eqn:lemma1_second}, we have
    \begin{align}
      \|\qhstarsqmones\|^2-\frac{\qe}{2}&\geq (1-q_{r4})\left(1+\frac{1}{2}\|\Bar{t}^B\|^2\right)-\frac{\|\Bar{q}_r\cdot \Bar{t}^B\|^2}{8}\nonumber\\
      &\geq (1-q_{r4})\left(1+\frac{1}{2}\|\Bar{t}^B\|^2\right)-\frac{(1-q^2_{r4})^2 \|\Bar{t}^B\|^2}{8}\nonumber\\
      &\geq(1-q_{r4})\left(1+\frac{1}{4}\|\Bar{t}^B\|^2\right)\geq 0\nonumber
    \end{align}
    Hence the result follows.
\end{proof}
We now state the first main result of this paper.
\begin{theorem}
    \normalfont ($\textbf{Main result 1: SGES}$) Consider the tracking error dynamics given by \eqref{eqn:system_dynamics}. Let the feedback control law \eqref{eqn:feedback_control_law} be applied to \eqref{eqn:system_dynamics}. Then, $\bx:=[\qhbd-\mathbf{1},\;\ohbd]=0$ is semi-globally exponentially stable equilibrium point for the closed-loop dynamics \eqref{eqn:system_dynamics}, i.e
\begin{align}
    \| \bx(t) \|\leq m(R)e^{-\alpha(R)(t-t_0)}\|\bx(t_0)\|,
\end{align}
for all $\bx(t_0)\in\mathcal{B}_R$ and $t\geq t_0\geq 0$. The expressions for $m(R)$ and $\alpha(R)$ are given by
\begin{align}
    & \alpha(R)=\frac{\beta(R)k_d}{2J_M},\nonumber\\
    & m(R)=\frac{\sqrt{(1+R^2)e^{\frac{J_M}{2k_p}R^2}-1+\sqrt{\frac{3}{2}}(1+R)\frac{J_M}{2}R^2}}{R\sqrt{k}_1},\nonumber\\
    &    0<\beta(R)<\min\left\{\frac{c}{2k_p},\frac{k_p}{2(1+R)\left(k_d[{\frac{3}{2}}(1+R)^2+\frac{c}{J_M}]\right)},1\right\},\nonumber\\
    &    c\geq \max\left\{\frac{4k_pk_0}{k_d},\frac{3J_M(1+R)^2}{4},\frac{J_Mk_p}{J_m}\right\},\nonumber\\
    &   k_0=\left(\frac{9}{4}J_M(1+R)+\frac{(3k_d/2+9J_M\delta(1+R)/2)^2(1+R)^3}{k_p}\right)\nonumber
\end{align}
where $J_M$ and $J_m$ are the maximum and minimum eigenvalues of dual inertia matrix $J$ respectively.
\end{theorem}
\begin{proof}
 To prove SGES, we consider the following candidate Lyapunov function $V$
\begin{align}
    V=&c\left[\text{exp}\left(\frac{V_0}{k_p}\right)-1\right]+\qhstarsqmones\circ J\star\ohbd\nonumber\\
    = &c\left[\left(1+\qenorm\right)\text{exp}\left(\frac{1}{2k_p}\left(\ohbdb\right)^{\mathrm{s}} \circ\left(J\star\left(\ohbdb\right)^{\mathrm{s}}\right)\right)-1\right]+\qhstarsqmones\circ J\star\ohbd\nonumber
\end{align}
First, we show that $V$ is positive definite. To that end
\begin{align}
    V&=c\left[\left(1+\qenorm\right)\left(1+\frac{1}{2k_p}\left(\ohbdb\right)^{\mathrm{s}} \circ\left(J\star\left(\ohbdb\right)^{\mathrm{s}}\right)\right)-1\right]+\qhstarsqmones\circ J\star\ohbd\nonumber\\
    &\geq c\left[{\qe}+\frac{1}{2k_p}\left(\ohbdb\right)^{\mathrm{s}} \circ\left(J\star\left(\ohbdb\right)^{\mathrm{s}}\right)\right]-\frac{J_M}{2}\left[\|\qhstarsqmones\|^2+\|\ohbd\|^2\right]\nonumber\\
    &\geq \left({c}-\frac{J_M}{2}\sqrt{\frac{3}{2}}^2(1+R)^2\right)\|\qhbd-\mathbf{1}\|^2+\left(\frac{cJ_m}{2k_p}-\frac{J_M}{2}\right)\|\ohbd\|^2\nonumber
\end{align}
where for the last inequality we have used Lemma 1 from \cite{filipe2013rigid_tracking} and the fact that $\|\qhbd^\star\|\leq (1+R)$. The terms $J_m$ and $J_M$ denote the minimum and maximum eigenvalues of $J$.
Therefore, if $c>\frac{J_M}{2}\text{max}({3(1+R)^2/2},2k_p/J_m)$, then there exists $k_1>0$ such that
\begin{align}
    V\geq k_1(\qe+\|\ohbd\|^2)
    \label{eqn:radially_unbounded}
\end{align}
which ensures that $V$ is a radially unbounded function and belongs to $\mathcal{KR}$\footnote{A function $\phi:[0,\infty)\rightarrow\mathbb{R}^+$ with $\phi(0)=0$ is called a class $\mathcal{K}$ function if it is continuous and strictly increasing. In addition, a function $\phi$ is a class $\mathcal{KR}$ function, if it is a $\mathcal{K}$ function and $\underset{r\rightarrow\infty}{\lim}\;\phi(r)=\infty$.}. 
The time derivative of $V$ along the system trajectories governed by \eqref{eqn:system_dynamics} is given by
\begin{align}
\dot{V}=c\left[\text{exp}\left(\frac{V_0}{k_p}\right)\frac{\dot{V}_0}{k_p}\right]+\dot{N}
\end{align}

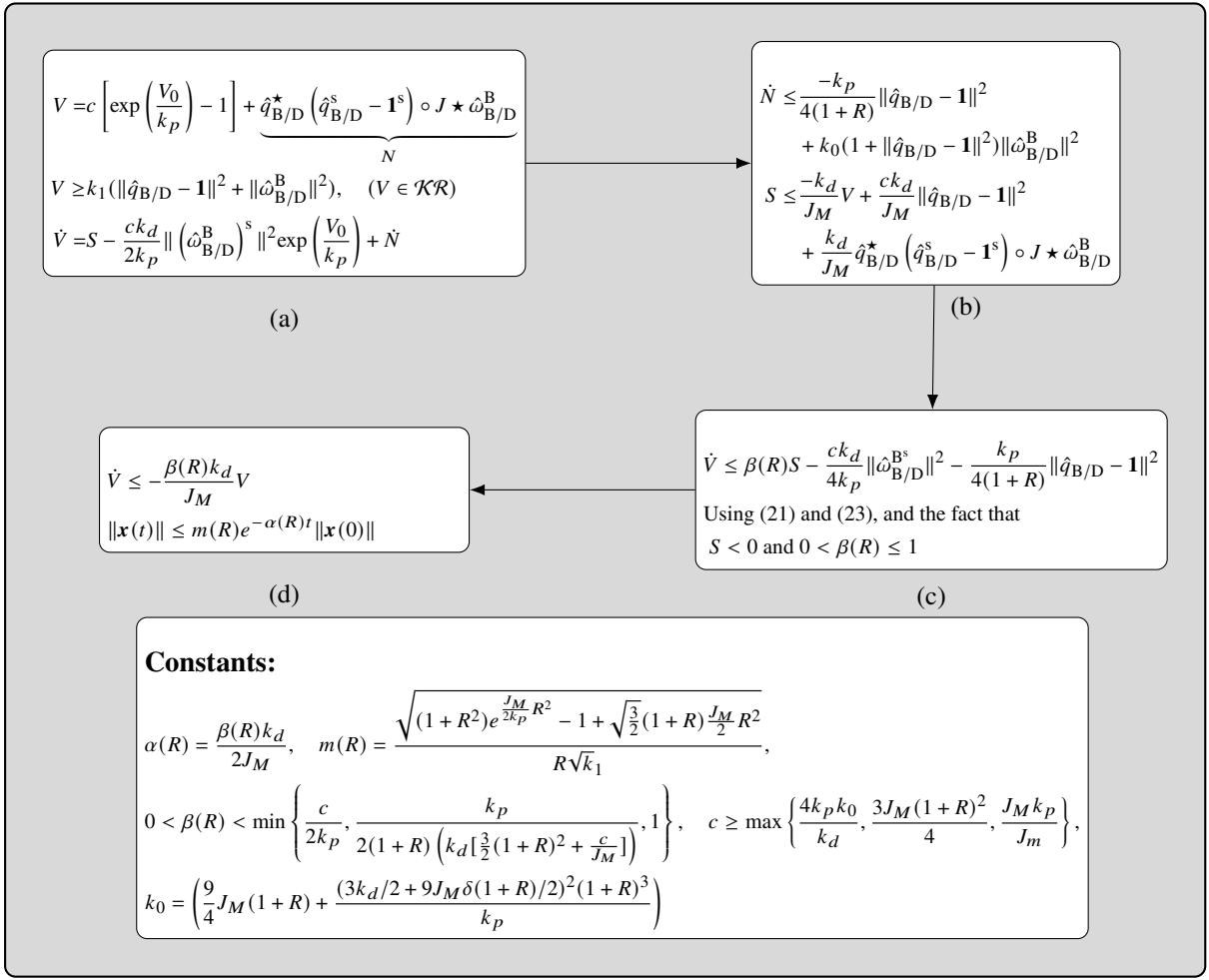
\begin{figure}[H]
\centering
\begin{tikzpicture}[block/.style={rectangle, draw, minimum height=1.5cm, minimum width=3cm,font=\small,fill=white,rounded corners,scale=0.9},
                    bigblock/.style={rectangle, draw, thick,rounded corners},
                    line/.style={-{Latex[length=2mm]}},
                    node distance=2cm and 3cm]

  \node[block] (a) {    \begin{minipage}{.2\textwidth}\begin{align}
   V=&c\left[\text{exp}\left(\frac{V_0}{k_p}\right)-1\right]+\underbrace{\qhstarsqmones\circ J\star\ohbd}_N\nonumber\\ 
       V\geq& k_1(\qe+\|\ohbd\|^2),\quad(V\in\mathcal{KR})\nonumber\\
   \dot{V}=&S-\frac{ck_d}{2k_p}\|\left(\ohbdb\right)^{\mathrm{s}}\|^2\text{exp}\left(\frac{V_0}{k_p}\right)+\dot{N}\nonumber
\end{align}
\end{minipage}
};
  \node[block, right=of a] (b) {\small\begin{minipage}{.2\textwidth}\begin{align}
        \dot{N}\leq&\frac{-k_p}{4(1+R)}\qenorm\nonumber\\
        &+k_0(1+\qenorm)\omegaenorm\nonumber\\
     S\leq& \frac{-k_d}{J_M}V+\frac{ck_d}{J_M}\qenorm\nonumber\\
     &+\frac{k_d}{J_M}\qhstarsqmones\circ J\star\hat{\omega}_{\mathrm{B} / \mathrm{D}}^{{\mathrm{B}}}\nonumber
\end{align}
\end{minipage}};
  \node[block, below=of a] (c) {
\begin{minipage}{.2\textwidth}\begin{align}
     &\dot{V}\leq-\frac{\beta(R) k_d}{J_M}V\nonumber\\
     &    \|\boldsymbol{x}(t)\|\leq m(R)e^{-\alpha(R)t}\|\boldsymbol{x}(0)\|\quad\quad\quad\quad\nonumber
\end{align}
\end{minipage}};
  \node[block, right=of c] (d) {\begin{minipage}{.2\textwidth}\begin{align}
     &\dot{V}\leq\beta(R) S-\frac{ck_d}{4k_p}\|\hat{\omega}_{\mathrm{B} / \mathrm{D}}^{{\mathrm{B}}^{\mathrm{s}}}\|^2-\frac{k_p}{4(1+R)}\qe\nonumber\\
     &\text{Using (21) and (23), and the fact that}\nonumber\\
     &\text{ $S<0$ and $0<\beta(R)\leq1$ }\quad\quad\nonumber
\end{align}
\end{minipage}};
  \node[block] (e) at ($(c.south east)+(1.9,-3)$) {\begin{minipage}{.2\textwidth}\begin{align}
  &\textbf{\large Constants:}\nonumber\\
    & \alpha(R)=\frac{\beta(R)k_d}{2J_M},\quad m(R)=\frac{\sqrt{(1+R^2)e^{\frac{J_M}{2k_p}R^2}-1+\sqrt{\frac{3}{2}}(1+R)\frac{J_M}{2}R^2}}{R\sqrt{k}_1},\nonumber\\
    &    0<\beta(R)<\min\left\{\frac{c}{2k_p},\frac{k_p}{2(1+R)\left(k_d[{\frac{3}{2}}(1+R)^2+\frac{c}{J_M}]\right)},1\right\},\quad    c\geq \max\left\{\frac{4k_pk_0}{k_d},\frac{3J_M(1+R)^2}{4},\frac{J_Mk_p}{J_m}\right\},\nonumber\\
    &   k_0=\left(\frac{9}{4}J_M(1+R)+\frac{(3k_d/2+9J_M\delta(1+R)/2)^2(1+R)^3}{k_p}\right)\nonumber
\end{align}
\end{minipage}};
  \node[below=0.3cm of a] {(a)};
  \node at ($(b.south)+(0.4,-0.3)$) {(b)};
  \node[below=0.3cm of c] {(d)};
  \node[below=0.1cm of d] {(c)};

  \draw[-{Latex[length=2mm]}] (a) -- (b) node[midway, above] {};
  \draw[-{Latex[length=2mm]}] (b) -- (d) node[midway, right] {};
  \draw[-{Latex[length=2mm]}] (d) -- (c) node[midway, below] {};

  \begin{scope}[on background layer]
    \node[bigblock, fit=(a) (b) (c) (d) (e), inner sep=0.5cm, fill=gray!30] (big) {};
  \end{scope}
  

\end{tikzpicture}
\caption{The figure illustrates the main steps involved in proving Theorem \ref{thm:es}. We begin by choosing a candidate Lyapunov function $V$, as demonstrated in Fig. (1a), and select a suitably large independent parameter $c$. This selection ensures that the function $V$ will be radially unbounded, i.e., $V\in\mathcal{KR}$.
Next, we leverage Lemma \ref{lemma:inequality_new} in conjunction with dual quaternion algebra to derive an upper limit for $\dot{N}$ and $S$ (Fig. (1b)). With the established condition that $S<0$, we then determine an independent parameter $0<\beta(R)< 1$, which is used to establish an upper boundary for $\dot{V}$, as depicted in Fig. (1c).
In our final step, as illustrated in Fig. (1d), we use the previously obtained upper bounds for $\dot{N}$ and $S$ to derive an upper limit for $\dot{V}$. This, in turn, permits us to establish the semi-global exponential stability, thus completing the proof of Theorem \ref{thm:es}.}
\label{fig:outline_of_the_proof}
\end{figure}






where $N=\qhstarsqmones\circ J\star\ohbd$. Now, using the fact that $\dot{V}_0=-k_d\|\left(\ohbdb\right)^{\mathrm{s}}\|^2$ (from Theorem \ref{thm:as}), we have
\begin{align}
    \dot{V}=&c\left[\frac{-k_d\|\left(\ohbdb\right)^{\mathrm{s}}\|^2}{k_p}\text{exp}\left(\frac{V_0}{k_p}\right) \right]+\dot{N}\nonumber\\
    = &c\left[\frac{-k_d\|\left(\ohbdb\right)^{\mathrm{s}}\|^2}{k_p}\left(1+{\qe}\right)\text{exp}\left(\frac{1}{2k_p}\left(\ohbdb\right)^{\mathrm{s}} \circ\left(J\star\left(\ohbdb\right)^{\mathrm{s}}\right)\right)\right]+\dot{N}\nonumber\\
    =&S-\frac{ck_d}{2k_p}\|\left(\ohbdb\right)^{\mathrm{s}}\|^2\text{exp}\left(\frac{V_0}{k_p}\right)+\dot{N}
    \label{eqn:vdot_upper_bound}
\end{align}
where $S$ is given by
\begin{align}
    S=\frac{-ck_d\|\left(\ohbdb\right)^{\mathrm{s}}\|^2}{2k_p}\left(1+{\qe }\right)\text{exp}\left(\frac{1}{2k_p}\left(\ohbdb\right)^{\mathrm{s}} \circ\left(J\star\left(\ohbdb\right)^{\mathrm{s}}\right)\right)
    \label{eqn:expression_S}
\end{align}
and $\dot{N}$ is given by
\begin{equation}
\begin{split}
    \dot{N}=&[\qhbddot^\star(\qhbds-\mathbf{1}^\mathrm{s})+\qhbd^\star\qhbddot^\mathrm{s}]\circ J\star\ohbdb+\qhstarsqmones\circ J\star{\dohbd}.\nonumber
\end{split}
\end{equation}

Substituting $\hat{f}^\mathrm{B}$ (which is present in the expression for $\dohbd$) from \eqref{eqn:feedback_control_law} in expression for $\dot{N}$, we have
\begin{align}
    \dot{N}=&[\qhbddot^\star(\qhbds-\mathbf{1}^\mathrm{s})+\qhbd^\star\qhbddot^\mathrm{s}]\circ J\star\ohbdb+\qhstarsqmones\circ\left[- k_{p}\frac{\left(\qhstarsqmones\right)}{1+\qenorm}-k_{d}\left(\ohbdb\right)^{\mathrm{s}}\right.\nonumber \\
&\left.+\ohdib \times\left(J\star\left(\ohdib\right)^{\mathrm{s}}\right)-\left(\ohbdb+\ohdib\right) \times\left(J\star\left(\left(\ohbdb\right)^{\mathrm{s}}+\left(\ohdib\right)^{\mathrm{s}}\right)\right)-J\star\left(\ohdib \times \ohbdb\right)^{\mathrm{s}}\right]\nonumber\\
=&[\qhbddot^\star(\qhbds-\mathbf{1}^\mathrm{s})+\qhbd^\star\qhbddot^\mathrm{s}]\circ J\star\ohbdb+\qhstarsqmones\circ\left[- {k_{p}\frac{\left(\qhstarsqmones\right)}{1+\qenorm}}-k_{d}\left(\ohbdb\right)^{\mathrm{s}}\right.\nonumber \\
&\left.-\left(\ohbdb\right) \times\left(J\star\left(\ohbdb\right)^{\mathrm{s}}\right)-\left(\ohdib\right) \times\left(J\star\left(\ohbdb\right)^{\mathrm{s}}\right)-\left(\ohbdb\right) \times\left(J\star\left(\ohdib\right)^{\mathrm{s}}\right]-J\star\left(\ohdib \times \ohbdb\right)^{\mathrm{s}}\right]\nonumber\\
\leq& \left[\frac{(\qhbd\ohbdb)^\star}{2}(\qhbds-\mathbf{1}^\mathrm{s})+\qhbd^{\star}\frac{(\qhbd\ohbd)^\mathrm{s}}{2}\right]\circ J\star\ohbdb-{k_{p}\frac{\|\qhstarsqmones\|^2}{1+\qenorm}}\nonumber\\
&-k_d\qhstarsqmones\circ\omegae+\frac{J_M}{2}\sqrt{\frac{3}{2}}(1+R)(1+\qenorm)\omegaenorm\nonumber\\
&\quad+3J_M\sqrt{\frac{3}{2}}(1+R)\delta\omegaenormsqrt\|\qhbdmone\|\nonumber
\label{eqn:N_dot_first}
\end{align}
where we have used Lemma 1 from  \cite{filipe2013rigid_tracking} and $\delta=\sup\;\|\ohbdi\|$ in the last step. 
Using Lemma \ref{lemma:inequality_new}, we have
\begin{align}
    -\|\qhstarsqmones\|^2\leq -\frac{\qe}{2}\nonumber
\end{align} 
 Further, using Lemma 1 from \cite{filipe2013rigid_tracking} and the fact that $\|\qhbd^\mathrm{s}-\mathbf{1}^\mathrm{s}\|\leq 1+\|\qhbd-\mathbf{1}\|^2$, we have
\begin{align}
    \dot{N}&\leq -\frac{k_p\qenorm}{2(1+R)}+\frac{3J_M(1+R)}{2}(1+\qenorm)\omegaenorm-k_d\qhstarsqmones\circ\omegae+\nonumber\\
    &\quad\frac{J_M}{2}\sqrt{\frac{3}{2}}(1+R)(1+\qenorm)\omegaenorm
    +\frac{9}{2}(1+R)J_M\delta(1+\qenorm)\|\qhbdmone\|\|\ohbdb\|\nonumber\\
    &\leq -\frac{k_p}{4(1+R)}\qenorm-\frac{k_p}{4(1+R)}\qe+\frac{9J_M(1+R)}{4}(1+\qenorm)\omegaenorm\nonumber\\
    &\quad +(3k_d/2+9J_M\delta(1+R)/2)(1+R)\|\qhbdmone\|\|\ohbdb\|\nonumber\\
    &\leq \frac{-k_p}{4(1+R)}\qenorm+\left[\frac{9}{4}J_M(1+R)+\frac{(3k_d/2+9J_M\delta(1+R)/2)^2(1+R)^3}{k_p}\right](1+\qenorm)\omegaenorm\nonumber\\
    &\leq\frac{-k_p}{4(1+R)}\qenorm+k_0(1+\qenorm)\omegaenorm
\end{align}
where $k_0$ is given by
\begin{align}
   k_0=\left(\frac{9}{4}J_M(1+R)+\frac{(3k_d/2+9J_M\delta(1+R)/2)^2(1+R)^3}{k_p}\right)\nonumber
\end{align}
Coming to expression of $S$ in \eqref{eqn:expression_S}, we have
\begin{align}
   S=&\frac{-ck_d}{k_p}\left(1+{\qe }\right)\text{exp}\left(\frac{1}{2k_p}\left(\ohbdb\right)^{\mathrm{s}} \circ\left(J\star\left(\ohbdb\right)^{\mathrm{s}}\right)\right)\left(\frac{\|(\ohbdb)^{\mathrm{s}}\|^2J_M}{2J_M}\right)\nonumber\\
   \leq& \frac{-ck_d}{J_M}\left(1+{\qe }\right)\text{exp}\left(\frac{1}{2k_p}\left(\ohbdb\right)^{\mathrm{s}} \circ\left(J\star\left(\ohbdb\right)^{\mathrm{s}}\right)\right)\left(\frac{\left(\ohbdb\right)^{\mathrm{s}} \circ\left(J\star\left(\ohbdb\right)^{\mathrm{s}}\right)}{2k_p}\right)
\end{align}
 Using the fact that $-xe^x<-(e^x-1)$ for all $x>0$, we have
\begin{align}
    S\leq&\frac{-ck_d}{J_M}\left(1+{\qe }\right)\left(\text{exp}\left(\frac{1}{2k_p}\left(\ohbdb\right)^{\mathrm{s}} \circ\left(J\star\left(\ohbdb\right)^{\mathrm{s}}\right)\right)-1\right)\nonumber\\
    =&\frac{-ck_d}{J_M}\left[\left(1+{\qe }\right)\text{exp}\left(\frac{1}{k_p}\left(\ohbdb\right)^{\mathrm{s}} \circ\left(J\star\left(\ohbdb\right)^{\mathrm{s}}\right)\right)\right]+\frac{ck_d}{J_M}\left(1+{\qe }\right)\nonumber\\
    \leq&\frac{-k_d}{J_M}\left[c\left[\left(1+{\qe }\right)\text{exp}\left(\frac{1}{2k_p}\left(\ohbdb\right)^{\mathrm{s}} \circ\left(J\star\left(\ohbdb\right)^{\mathrm{s}}\right)\right)-1\right]\right]-\frac{ck_d}{J_M}\nonumber\\
    -&\frac{k_d}{J_M}\qhstarsqmones\circ J\star\hat{\omega}_{\mathrm{B} / \mathrm{D}}^{{\mathrm{B}}}+\frac{ck_d}{J_M}\qenorm+\frac{ck_d}{J_M}+\frac{k_d}{J_M}\qhstarsqmones\circ J\star\hat{\omega}_{\mathrm{B} / \mathrm{D}}^{{\mathrm{B}}}\nonumber\\
    \leq&\frac{-k_d}{J_M}V+\frac{ck_d}{J_M}\qenorm+\frac{k_d}{J_M}\qhstarsqmones\circ J\star\hat{\omega}_{\mathrm{B} / \mathrm{D}}^{{\mathrm{B}}}
    \label{eqn:s_upper_bound}
\end{align}
Since $S$ is non-positive, we choose $\beta(R)\in(0,1)$ such that
\begin{align}
    \dot{V}\leq &\beta(R) S-\frac{ck_d}{2k_p}(1+\qe)\|\left(\ohbdb\right)^{\mathrm{s}}\|^2-\frac{k_p}{4(1+R)}\qe+k_0(1+\qe)\|\left(\ohbdb\right)^{\mathrm{s}}\|^2\nonumber\\
        \leq&\beta(R) S-\frac{ck_d}{4k_p}(1+\qe)\|\left(\ohbdb\right)^{\mathrm{s}}\|^2-\frac{k_p}{4(1+R)}\qe\nonumber\\
        &-\left(\frac{ck_d}{4k_p}-k_0\right)(1+\qe)\|\left(\ohbdb\right)^{\mathrm{s}}\|^2\nonumber\\
    \leq&\beta(R) S-\frac{ck_d}{4k_p}\|\left(\ohbdb\right)^{\mathrm{s}}\|^2-\frac{k_p}{4(1+R)}\qe-\left(\frac{ck_d}{4k_p}-k_0\right)(1+\qe)\|\left(\ohbdb\right)^{\mathrm{s}}\|^2\nonumber\\
    \leq &\beta(R) S-\frac{ck_d}{4k_p}\|\left(\ohbdb\right)^{\mathrm{s}}\|^2-\frac{k_p}{4(1+R)}\qe
\end{align}
where $c$ is chosen such that
\begin{align}
    c\geq \max\left\{\frac{4k_pk_0}{k_d},\frac{3J_M(1+R)^2}{4},\frac{J_Mk_p}{J_m}\right\}
\end{align}
Using \eqref{eqn:s_upper_bound}, we have
\begin{align}
    \dot{V}\leq&\beta(R)\left(\frac{-k_d}{J_M}V+\frac{ck_d}{J_M}\qenorm+\frac{k_d}{J_M}\qhstarsqmones\circ J\star\hat{\omega}_{\mathrm{B} / \mathrm{D}}^{{\mathrm{B}}}\right)-\frac{ck_d}{4k_p}\|\left(\ohbdb\right)^{\mathrm{s}}\|^2\nonumber\\
    &-\frac{k_p}{4(1+R)}\qe\nonumber\\
    \leq& -\frac{\beta(R) k_dV}{J_M}-\left(\frac{k_p}{4(1+R)}-\frac{\beta(R) c k_d}{J_M}\right)\qe-\frac{ck_d}{4k_p}\|\left(\ohbdb\right)^{\mathrm{s}}\|^2\nonumber\\
    &+\frac{\beta(R) k_d}{J_M}\qhstarsqmones\circ J\star\left(\ohbdb\right)^{\mathrm{s}}\nonumber\\
     \leq& -\frac{\beta(R) k_dV}{J_M}-\left(\frac{k_p}{4(1+R)}-\frac{\beta(R) c k_d}{J_M}\right)\qe-\frac{ck_d}{4k_p}\|\left(\ohbdb\right)^{\mathrm{s}}\|^2\nonumber\\
     &+\frac{\beta(R) k_d}{2}\left(\|\qhstarsqmones\|^2+\|\left(\ohbdb\right)^{\mathrm{s}}\|^2 \right)\nonumber\\
     \leq & -\frac{\beta(R) k_d}{J_M}V+\left(\frac{\beta(R) k_d}{2}\left({\frac{3}{2}}(1+R)^2+\frac{c}{J_M}\right)-\frac{k_p}{4(1+R)}\right)\qe+\frac{k_d}{2}\left(\beta(R)-\frac{c}{2k_p}\right)\|\left(\ohbdb\right)^{\mathrm{s}}\|^2\nonumber\\
     \leq&-\frac{\beta(R) k_d}{J_M}V
     \label{eqn:lyapunov_final}
\end{align}
where $\beta(R)$ is chosen such that
\begin{align}
    0<\beta(R)<\min\left\{\frac{c}{2k_p},\frac{k_p}{2(1+R)\left(k_d[{\frac{3}{2}}(1+R)^2+\frac{c}{J_M}]\right)},1\right\}
\end{align}
Therefore, by comparison lemma, we have
\begin{align}
    V(t)\leq V(0)\text{exp}\left(-\frac{\beta(R) k_dt}
    {J_M}\right),\quad\forall t\geq 0
    \label{eqn:Vt_less_than_V0_exp}
\end{align}
Clearly from \eqref{eqn:radially_unbounded}, $V(t)$ is lower bounded by $k_1\|\bx(t)\|^2$. For $\bx(0)\neq 0$, we have
\begin{align}
    V(0)=\frac{V(0)}{\|\bx(0)\|^2}\|\bx(0)\|^2\leq \underbrace{\frac{\left(\sqrt{(1+\|\bx(0)\|^2)e^{\frac{J_M}{2k_p}\|\bx(0)\|^2}-1+\sqrt{\frac{3}{2}}(1+R)\frac{J_M}{2}\|\bx(0)\|^2}\right)^2}{\|\bx(0)\|^2}}_{:=L(\|\bx(0)\|)}\|\bx(0)\|^2
        \label{eqn:v_0_upper_bound}
        \end{align}
        It can be easily be shown that $L(\|\bx(0)\|)$ is an increasing function for all $\|\bx(0)\|>0$. Since $\bx(0)\in\mathcal{B}_R$, \eqref{eqn:v_0_upper_bound} becomes
        \begin{align}
  V(0)  \leq \frac{\left(\sqrt{(1+R^2)e^{\frac{J_M}{2k_p}R^2}-1+\sqrt{\frac{3}{2}}\frac{J_M}{2}(1+R)R^2}\right)^2}{R^2}\|\bx(0)\|^2
    \label{eqn:v_0_upper_bound_1}
\end{align}
Using \eqref{eqn:v_0_upper_bound_1} and \eqref{eqn:radially_unbounded}, \eqref{eqn:Vt_less_than_V0_exp} becomes
\begin{align}
    k_1\|\bx(t)\|^2\leq V(t)\leq \frac{\left(\sqrt{(1+R^2)e^{\frac{J_M}{2k_p}R^2}-1+\sqrt{\frac{3}{2}}\frac{J_M}{2}(1+R)R^2}\right)^2}{R^2}\text{exp}\left(-\frac{\beta(R) k_dt}
    {J_M}\right)\|\bx(0)\|^2
    \label{eqn:Vt_less_than_x0}
\end{align}
Consequently, \eqref{eqn:Vt_less_than_x0} becomes
\begin{align}
    \|\bx(t)\|\leq m(R)e^{-\alpha(R)t}\|\bx(0)\|
\end{align}
where $\alpha(R)=\frac{\beta(R)k_d}{2J_M}$ and $m(R)=\frac{\sqrt{(1+R^2)e^{\frac{J_M}{2k_p}R^2}-1+\sqrt{\frac{3}{2}}(1+R)\frac{J_M}{2}R^2}}{R\sqrt{k}_1}$. From Definition \ref{thm:es}, we  conclude that $\bx=0$ is a semi-global exponentially stable equilibrium point for the closed-loop system \eqref{eqn:system_dynamics} under feedback control law \eqref{eqn:feedback_control_law}.
\end{proof}
In the following section, we discuss the relation of the proposed feedback controller \eqref{eqn:feedback_control_law} with optimal control.
{
\subsection{Relation of proposed feedback control law to optimal control}
\begin{theorem}
    \normalfont [Theorem 8.2, \cite{haddad2008nonlinear_wassim}] Consider the nonlinear controlled dynamical system $\dot{\bx}=F(\bx,\bu)$ with performance functional
\begin{align}
J\left(\bx_0, \bu\right) \triangleq \int_0^{\infty} L(\bx(t), \bu(t)) \mathrm{dt},\nonumber
\end{align}
 Assume that there exists a continuously differentiable function $V: \mathcal{B}_R \rightarrow \mathbb{R}$ and a control law $\phi$ such that
 \begin{subequations}
    \begin{align}
& V(0)=0, \quad V(\bx)>0, \quad\phi(0)=0, \quad V^{\prime}(\bx) F(\bx, \phi(\bx))<0\quad \bx \in \mathcal{B}_R, \quad \bx \neq 0 \\
& H(\bx, \phi(\bx))=0,\quad H(\bx, \bu) \geq 0, \quad \bx \in \mathcal{B}_R
\end{align}
\label{eqn:V_conditions_optimal}
 \end{subequations}
where $H(\bx, \bu) \triangleq L(\bx, \bu)+V^{\prime}(\bx) F(\bx, \bu)$.
In addition, if $
\bx_0 \in \mathcal{B}_R$ then the feedback control $\bu=\phi(\bx)$ minimizes $J\left(\bx_0, \bu\right)$ in the sense that
\begin{align}
J\left(\bx_0, \phi(\bx(\cdot))\right)=\min _{\bu \in \mathbb{R}^m} J\left(\bx_0, \bu\right)\nonumber
\end{align}
\label{thm:wassim}
\end{theorem}
As a consequence of Theorem \ref{thm:wassim}, we have the following result
\begin{theorem}
    \normalfont For the system \ref{eqn:system_dynamics}, the feedback controller \ref{eqn:feedback_control_law} is the solution to the following optimal control problem i.e.
    \begin{align}
\hat{f}^\mathrm{B}=\underset{\bu\in\mathbb{R}^m}{\text{argmin}} \;\int_0^{\infty} L(\bx(t), \bu(t)) \mathrm{dt}\nonumber
    \end{align}
    where $ L(\bx(t), \bu(t))=-\dot{V}\left(:=\frac{-\partial V(\bx)}{\partial 
    \bx}\dot{\bx}\right)>0$ and $\bx:=[\qhbd-\mathbf{1},\;\ohbd]$.
\end{theorem}
\begin{proof}
  Choose  $L(\bx(t), \bu(t))=-\dot{V}$, which implies $H(\bx,\bu)=H(\bx,\hat{f}^\mathrm{B})=0$. Furthermore, since $V(0)=0$, $V(\bx)>0$, $\dot{V}<0$, and $\hat{f}^\mathrm{B}=0$, from Theorem \ref{thm:wassim}, we can conclude that the feedback controller \ref{eqn:feedback_control_law} is the solution to the following optimal control problem i.e. $\hat{f}^\mathrm{B}=\underset{\bu\in\mathbb{R}^m}{\text{argmin}} \;\int_0^{\infty} L(\bx(t), \bu(t)) \mathrm{dt}$.
\end{proof}
}

\section{Robustness analysis\label{sec:robustness_analysis}}
In this section, we define and prove the local Input-to-State stability (ISS) for the system \eqref{eqn:system_dynamics} in the presence of time-varying additive bounded external disturbances. The ISS property provides a mechanism to ensure that for any bounded input, such as external disturbances or changes in system parameters, the state of the system remains bounded, ensuring the overall system stability. This is particularly important in aerospace applications where loss of stability can have catastrophic consequences. ISS not only assures a bounded response but also provides a measure of the degree of this boundedness, thereby providing a quantifiable measure of robustness. Formally, ISS is defined as follows:
\begin{definition}
\normalfont \textbf{(Local Input-to-State stability  (ISS)) \cite{khalil2002nonlinear}} 
 The system $\dot{\bx}=f(\bx,\bu)$ is said to be locally input-to-state stable (ISS), if there exist a class $\mathcal{KL}$\footnote{A continuous function $\beta:[0,r)\times [0,\infty)\rightarrow [0,\infty)$ is a class $\mathcal{KL}$ function if it satisfies the following three conditions. First, for a fixed $y$, $\beta(x,y)\in\mathcal{K}$ with respect to $x$. Second, $\beta(x,y)$ is decreasing with $y$. Lastly, $\underset{y\rightarrow\infty}{\lim}\;\beta(x,y)=0$.} function $\beta$, a class $\mathcal{K}$ function $\gamma$, and positive constants $k_1$ and $k_2$ such that for any initial state $\bx\left(t_0\right)$ with $\left\|\bx\left(t_0\right)\right\|<k_1$ and any input $\bu(t)$ with $\sup _{t \geq t_0}\|\bu(t)\|<k_2$, the solution $\bx(t)$ exists and satisfies
$$
\|\bx(t)\| \leq \beta\left(\left\|\bx\left(t_0\right)\right\|, t-t_0\right)+\gamma\left(\sup _{t_0 \leq \tau \leq t}\|\bu(\tau)\|\right)
$$
for all $t \geq t_0 \geq 0$. 
\end{definition}
Note that, Input-to-state stability (ISS) is particularly important in the presence of time-varying additive disturbances because it provides a framework to assess how external perturbations affect the system's state. Aerospace systems are often subjected to disturbances such as gravitational perturbations, solar radiation pressure, magnetic field interactions etc. These disturbances are inherently time-varying and can significantly impact the trajectory and attitude of the spacecraft. ISS ensures that if the disturbances are bounded, the system's state will remain bounded and converge to an acceptable range, thus preserving the system's functionality and preventing instability. This robustness is crucial for the design and operation of aerospace systems where maintaining performance in the face of such uncertainties is essential for safety, reliability, and effectiveness. The following theorem allows us to claim ISS using Lyapunov analysis.
\begin{theorem}
  \normalfont \textbf{Local ISS Theorem \cite{khalil2002nonlinear}}:  Let $D=\left\{\bx \in \mathbb{R}^n \mid\|\bx\|<r\right\}, D_u=\left\{\bu \in \mathbb{R}^m \mid\|\bu\|<r_u\right\}$, and $f:[0, \infty) \times D \times D_u \rightarrow \mathbb{R}^n$ be piecewise continuous in $t$ and locally Lipschitz in $\bx$ and $\bu$. Let $V:[0, \infty) \times D \rightarrow R$ be a continuously differentiable function such that
\begin{align*}
& \alpha_1(\|\bx\|)  \leq V(t, \bx) \leq \alpha_2(\|\bx\|) \\
& \frac{\partial V}{\partial t}+\frac{\partial V}{\partial x} f(t, \bx, \bu)  \leq-\alpha_3(\|\bx\|), \quad \forall (\bx, \bu)~~\text{s.t.}~~\|\bx\| \geq \rho(\|\bu\|)>0
\end{align*}
and $(t, \bx, \bu) \in[0, \infty) \times D \times D_u$ where $\alpha_1, \alpha_2, \alpha_3$ and $\rho$ are class $\mathcal{K}$ functions. Then, the nonlinear system $\dot{\bx}=f(t,\bx,\bu)$ is locally input-to-state stable with $\gamma=\alpha_1^{-1} \circ \alpha_2 \circ \rho\; , k_1=\alpha_2^{-1}\left(\alpha_1(r)\right)$, and $k_2=\rho^{-1}\left(\min \left\{k_1, \rho\left(r_u\right)\right\}\right)$. 
\label{thm:local_iss}
\end{theorem}
We now state the second main result of this paper.
\begin{theorem}
    \normalfont ($\textbf{Main result 2: Local ISS}$) Consider the dual quaternion-based rigid body dynamics with additive bounded disturbances $d_1(t),\;d_2(t)\in\mathcal{D}$ for all $t\geq 0$ where $\mathcal{D}$ is a compact set as follows:
\begin{align}
{\dqhbd}&=\frac{1}{2}\qhbd\ohbdb+d_1(t)\nonumber\\
\left(\dot{\hat{\omega}}_{\mathrm{B} / \mathrm{D}}^{\mathrm{B}}\right)^{\mathrm{s}}&=\left(J\right)^{-1} \star\left(\hat{f}^\mathrm{B}-\left(\ohbdb+\ohdib\right) \times\left(J\star\left(\left(\ohbdb\right)^{\mathrm{s}}+\left(\ohdib\right)^{\mathrm{s}}\right)\right)\right)\nonumber \\
&\quad\quad\left.-J\star\left(\qhbd^{\star} \dot{\hat{\omega}}_{\mathrm{D} / \mathrm{I}}^{\mathrm{D}} \qhbd\right)^{\mathrm{s}}-J\star\left(\ohdib \times \ohbdb\right)^{\mathrm{s}}+d_2(t)\right)\label{eqn:system_dynamics_disturbances}
\end{align} 
The closed-loop system \eqref{eqn:system_dynamics_disturbances} under the feedback control law \eqref{eqn:feedback_control_law} is local ISS.
\end{theorem}
\begin{proof}
Following similar steps as in the proof of Theorem \ref{thm:es}, it can be shown that $\dot{V}$ along system trajectories \eqref{eqn:system_dynamics_disturbances} turns out to be the following
\begin{align}
   \dot{V}\leq& -\frac{\beta k_d}{2J_M}V+\underbrace{\sqrt{\frac{3}{2}}(1+R)\|\qhbdmone\|d_m-\frac{k_p}{4(1+R)} \qe}_{\text{Term 1}}\nonumber\\
   &+\underbrace{(1-\beta(R))S+(1+2R)\|J\|\|\ohbdb\|d_m+\frac{c(1+k_p)}{k_p}\text{exp}\left(\frac{V_0}{2k_p}\right)\|\ohbdb\|d_m}_{\text{Term 2}}\nonumber
\end{align}
where $d_m=\underset{\mathcal{D}}{\max}\{d_1(t),\;d_2(t)\}$. To ensure that $ \dot{V}\leq -\frac{\beta k_d}{2J_M}V$, the following inequalities must be satisfied
\begin{subequations}
\begin{equation}
   \|\qhbdmone\| \geq \frac{4(1+R)^2\sqrt{3}d_m}{\sqrt{2}k_p},\quad\;\;\quad \text{(Term 1)}>0
   \label{eqn:negative1}
   \end{equation}
   \begin{equation}
   \|\ohbdb\|\geq \frac{\left[\|J\|(1+2R)+\frac{c(1+k_p)}{k_p}\text{exp}\left(\frac{V_0}{2k_p}\right)\right]d_m}{\frac{(1-\beta(R))ck_d}{2k_p}\text{exp}\left(\frac{V_0}{2k_p}\right)}\geq \frac{2(1+k_p)d_m}{(1-\beta(R))k_d},\quad\;\; \text{(Term 2)}>0
   \label{eqn:negative2}
\end{equation}
\label{eqn:negative}
\end{subequations}
Conditions \eqref{eqn:negative1} and \eqref{eqn:negative2} imply that $\|\qhbdmone\|d_m-\frac{k_p}{4(1+R)} \qe\leq 0$ and $(1-\beta(R))S+(1+2R)\|J\|.\|\ohbdb\|d_m+\frac{c(1+k_p)}{k_p}\text{exp}\left(\frac{V_0}{2k_p}\right)\|\ohbdb\|d_m\leq 0$ respectively. Note that the disturbance in the actuator can be considered as equivalent to the disturbance added to the $\ohbdb$ dynamics. Therefore, in this section, we do not analyze the case with actuator disturbance. If $\gamma$ is selected such that
\begin{align}
    \psi=\sqrt{\left( \frac{4(1+R)^2\sqrt{3}}{\sqrt{2}k_p}\right)^2+\left(\frac{2(1+k_p)}{(1-\beta(R))k_d}\right)^2}
    \label{eqn:psi_expression}
\end{align}
then from \eqref{eqn:negative},  we have $\|\bx\|\geq\psi d_m$. Therefore if $\|\bx\|\geq\psi d_m$, $ \dot{V}\leq -\frac{\beta k_d}{2J_M}V$ holds true. Further, since $V\in\mathcal{KR}$ (from \eqref{eqn:radially_unbounded}), we can conclude from Theorem \ref{thm:local_iss} that the closed-loop control law with time-varying additive bounded external disturbances is local Input-to-State Stable (ISS).
\end{proof}
{
\begin{figure}[H]
 \captionsetup[subfigure]{justification=centering}
 \centering
{\includegraphics[scale=0.38]{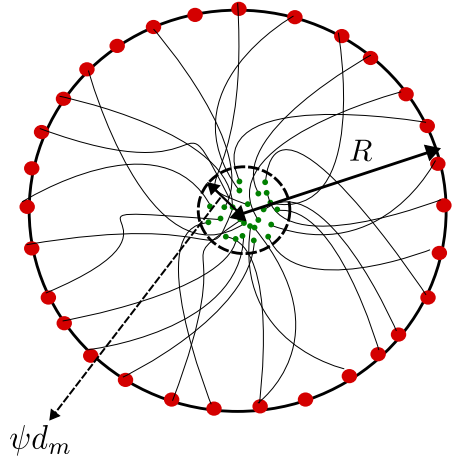}}
 \caption{Schematic diagram of trajectories under the proposed feedback controller \eqref{eqn:feedback_control_law} converging to a ball of radius equal or at most $\psi d_m$}
\label{fig:iss_diagram}
\end{figure}
}
\section{Motion and Safety constraints\label{sec:safety_and_motion_constraints}}
In this section, we leverage the notion of Control Barrier Functions (CBFs) \cite{ames2019control} to design a control framework that ensures that motion and safety constraints are satisfied for the rigid body. CBFs offer a simplified and compact method for ensuring safety and satisfying motion constraints in nonlinear controlled systems. This is particularly crucial in critical aerospace applications where safety is paramount. In such applications, even minor deviations from the desired motion trajectory or state can lead to catastrophic outcomes. The CBF based approach for synthesising control inputs basically modifies the nominal control input $\bu_0$ (defined by \eqref{eqn:feedback_control_law} in our case) as necessary to enforce the motion and safety constraints, while still trying to achieve the performance objectives as closely as possible. In essence, the nominal controller sets the performance targets, and the CBF-based controller ensures safety while striving to meet these targets. We now discuss the notion of CBF.

\subsection{Control Barrier Functions (CBFs)}
Let $h(\bx)$ be a scalar valued continuously differentiable function that satisfies the following:
\begin{align}
    h(\bx)>0\;\;\;\forall\; \bx\in\text{int}(\mathcal{S}),\quad   h(\bx)=0\;\;\;\forall\; \bx\in\partial \mathcal{S},\quad   h(\bx)<0\;\;\;\forall \;\bx\in\mathcal{X}\setminus\mathcal{S}
    \label{eqn:h_function}
\end{align}
where $\text{int}(\mathcal{S})$ and $\partial\mathcal{S}$ denotes the  interior and the boundary of the safe set $\mathcal{S}$ respectively. The system has a relative degree $r$ (positive integer) if, in
some neighborhood of the origin,
\begin{align}
    L_gL_f^{i-1}h(\bx)=0,\quad  L_gL_f^{r-1}h(\bx)\neq 0,\quad\forall\quad i=1,2,\dots,(r-1)\nonumber
\end{align}
The CBF is defined as follows
\begin{definition}
    \normalfont Given a control-affine dynamical system represented by $\dot{\bx}=f(\bx)+g(\bx)\bu$, $ h(\bx)$ defined by \eqref{eqn:h_function} is said to be a CBF if there exist a control input $\bu$ such that the following holds true
    \begin{align}
\underset{\bu\in\mathcal{U}}{\sup}\;L_f^{r}h(\bx)+L_gL_f^{r-1}h(\bx)\bu\geq -\boldsymbol{a}^\mathrm{T}\eta(h(\bx))
    \end{align}
 where $L_fh(\bx):=\frac{\partial h(\bx)}{\partial\bx}f(\bx)$ and $L_gh(\bx):=\frac{\partial h(\bx)}{\partial\bx}g(\bx)$ are the Lie derivatives of $h(\bx)$ along $f$ and $g$ respectively, $\eta(h(\bx))=\left[L_f^{r-1} h, L_f^{r-2} h, \ldots, h\right]^{\mathrm{T}}$, and $\boldsymbol{a}=\left[a_1, \ldots, a_r\right]^{\top} \in \mathbb{R}^r$ are the set of parameters chosen such that the roots of $\lambda^r+a_1 \lambda^{r-1}+\ldots+a_{r-1} \lambda+a_r=0$ are all
negative reals.
\end{definition}
Therefore, a feedback control input $\bu_s(\bx)$ that guarantees satisfaction of motion and safety constraints can be synthesized by solving the following (constrained) Quadratic Program (QP) as follows:
\begin{subequations}
\begin{align}
&\bu_s(\bx)=\underset{\bu\in\mathcal{U}}{\text{argmin}}\;\|\bu-\bu_0\|^2\\
    &  L_f^{r}h(\bx)+L_gL_f^{r-1}h(\bx)\bu\geq -\boldsymbol{a}^\mathrm{T}\eta(h(\bx))
\end{align}
    \label{eqn:cbf_qp_general}
\end{subequations}

where $\bu_0$ is the nominal controller. Note that for any fixed $\bx$, \eqref{eqn:cbf_qp_general} is a QP with $\bu$ as the optimization variable. The following subsection discusses the application of CBF for generating safe control inputs for spacecraft rendezvous and docking problem.

\subsection{Safe spacecraft rendezvous and docking}
We consider the final phase of the spacecraft rendezvous and docking problem where the chaser spacecraft docks in the target spacecraft while satisfying motion constraints such as approach path and input constraints. As shown in Fig. \ref{fig:approach_path_constraint}, the approach path constraint ensures that the vector $\boldsymbol{p}$ makes an angle of less than $\theta$ with the vector pointing from the chaser to the target spacecraft. For the spacecraft docking problem that captures the approach path constraint and the input constraints, we define the CBF function $h$ as follows:
\begin{align}
h({r}_{{\mathrm{B}/\mathrm{D}}})=\left\{\begin{array}{cc}
\frac{\bar{x}^3 \tan ^2(\theta)}{2 \bar{r}_1-\bar{x}}-\bar{y}^2-\bar{z}^2 & \bar{x} \in\left[0, \bar{r}_1\right] \\
\bar{x}^2+\bar{y}^2+\bar{z}^2-\bar{r}_3^2 & \bar{x} \in\left[-\bar{r}_3, \bar{r}_2\right] \\
\bar{x}^2+\bar{y}^2+\bar{z}^2+r^\star \sin \left(\pi \frac{\bar{x}}{\bar{r}_2}\right)\left(\bar{y}^2+\bar{z}^2\right)-\bar{r}_3^2 & \bar{x} \in\left[\bar{r}_1, \bar{r}_2\right]
\end{array}\right.
\label{eqn:cbf_defn}
\end{align}
where ${r}_{{\mathrm{B}/\mathrm{D}}}=[\Bar{x},\;\Bar{y},\;\Bar{z}]^\mathrm{T}$, $ \bar{r}_3>\bar{r}_2>\bar{r}_1>0,\;\theta>0,\;r^\star>0$ are constants. The value of $r^\star$ is chosen such that $h$ is continuous and differentiable at $x=\bar{r}_2$. To that end, $r^\star$ is given by
\begin{align}
r^\star=\frac{\bar{r}_3^2-\bar{r}_1^2\left[1+\tan ^2(\theta)\right]}{\bar{r}_1^2 \sin \left(\pi\left(\bar{r}_1 / \bar{r}_2\right)\right) \tan ^2(\theta)}\nonumber
\end{align}
The mathematical equation \eqref{eqn:cbf_defn} presents a composite surface ($h=0$), which comprises of three sections. The initial section is a curve based on the cissoid \cite{zhang2010ellipse_cissoid} and serves to delineate the ultimate approach corridor. Adjusting $\bar{r}_1$ and $\theta$ allows for modifications to the length and the angle included in the approach corridor, respectively. Although the second section consists of a partial spherical surface, its radius ($\bar{r}_3$) must be sufficiently large to encompass the target's components. Fig. \ref{fig:cbf_yellow} shows the surface $h=0$. The purpose of the third section is to link the preceding two sections. The control inputs are then synthesized by solving the following QP:
\begin{align}
& \bu_s({r}_{{\mathrm{B}/\mathrm{D}}})=  \underset{\bu\in\mathcal{U}}{\text{argmin}}\;\|\bu-{f}^\mathrm{B}\|^2\nonumber\\
    &\frac{\partial^2 h({r}_{{\mathrm{B}/\mathrm{D}}})}{\partial {r}^2_{{\mathrm{B}/\mathrm{D}}}}\ddot{{r}}_{{\mathrm{B}/\mathrm{D}}}\geq -\alpha(h({r}_{{\mathrm{B}/\mathrm{D}}}))
\end{align}
where ${f}^\mathrm{B}$ is the nominal force (which is extracted from the $\hat{f}^\mathrm{B}=(f^\mathrm{B},\tau^\mathrm{B})$ in \eqref{eqn:feedback_control_law}) applied by the chaser spacecraft that guarantees semi-global exponential stability for the translation motion. Note that the input constraints (for instance ball or box constraints for the spacecraft's thrust) can be incorporated in the set $\mathcal{U}$ and the translational dynamics is of relative degree 2. Furthermore, in this formulation, since we have only considered the approach path and input constraints, we do not have to modify the torque applied by the chaser spacecraft.
\begin{figure}[H]
 \captionsetup[subfigure]{justification=centering}
 \centering
 \begin{subfigure}{0.42\textwidth}
{\includegraphics[scale=0.18]{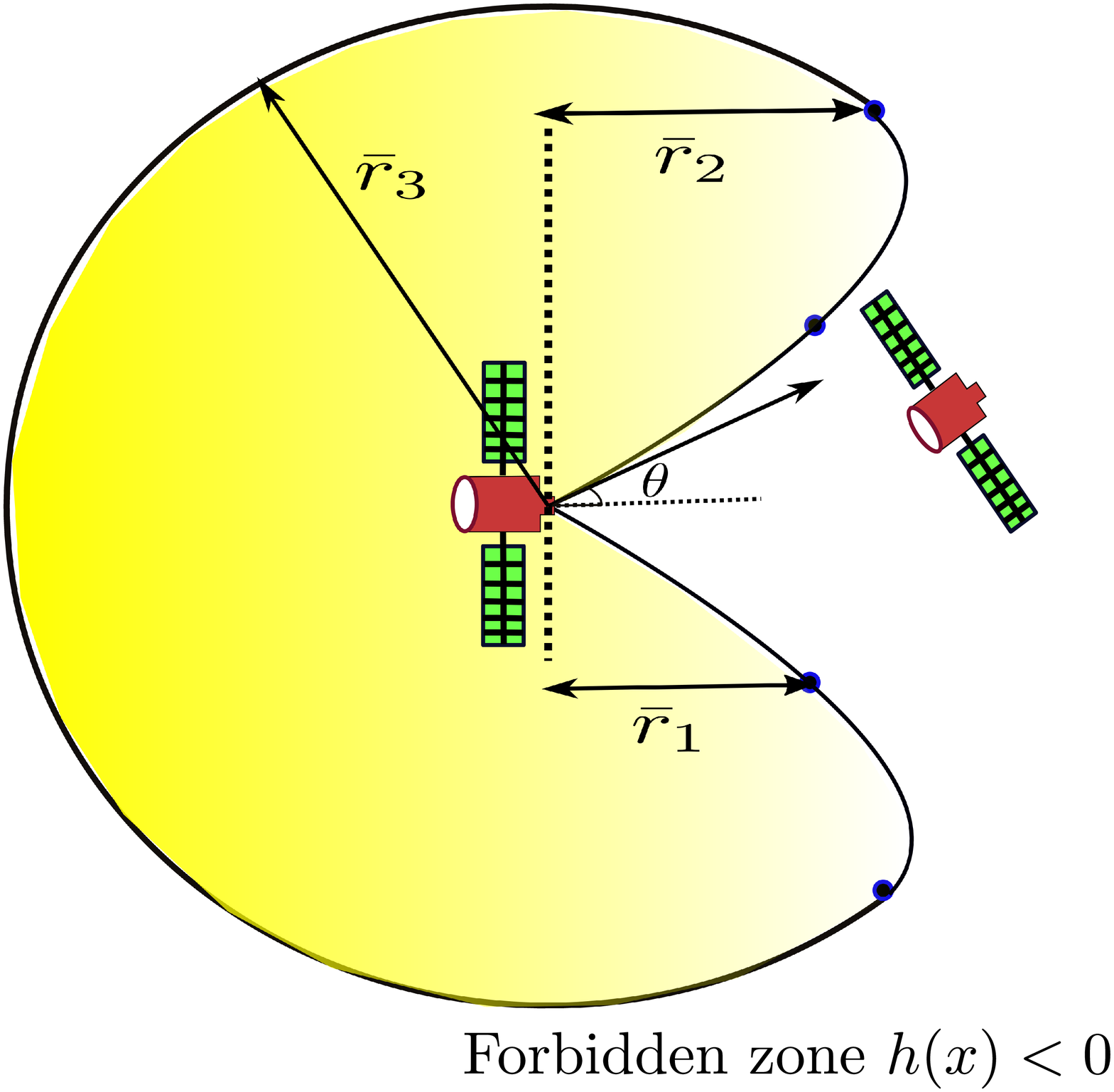}}
 \caption{Depiction of CBF $h(\bx)$}
\label{fig:cbf_yellow}
 \end{subfigure}
\label{fig:cbf_svg_yello}
 \begin{subfigure}{0.3\textwidth}
{\includegraphics[scale=0.11]{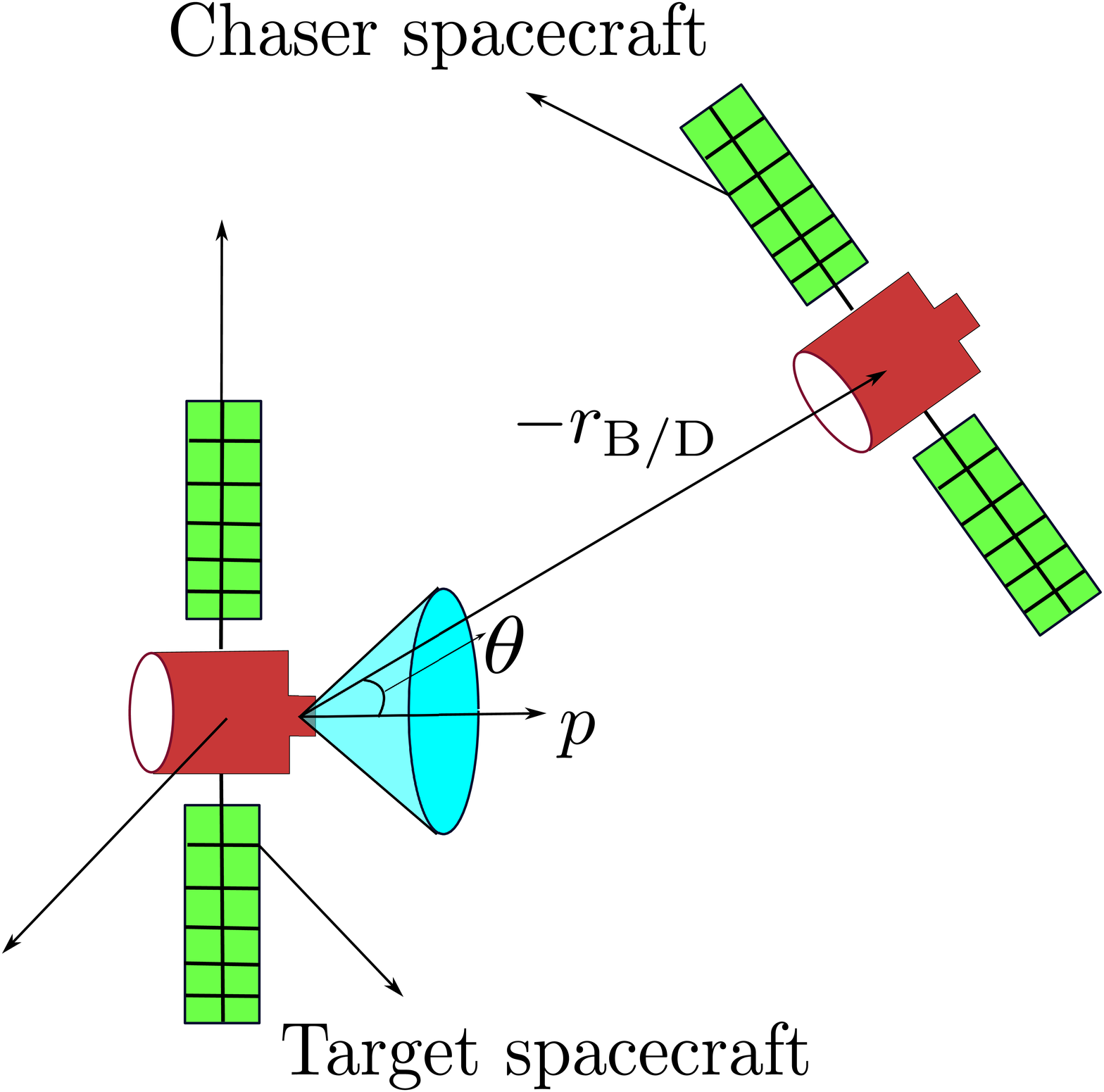}}
\caption{Approach path constraint}
\label{fig:approach_path_constraint}
\end{subfigure}
  \caption{}
\end{figure}
\begin{remark}
    \normalfont Recent methods \cite{breeden2021robust_cbf_rendezvous_1, breeden2021guaranteed_cbf_rendezvous_2} that use CBF for safe spacecraft rendezvous employ PID controllers as the nominal controller. In contrast, the nominal controller proposed herein, which is denoted as $f^\mathrm{B}$, is an exponentially stabilizing controller. This type of controller provides significant advantages by ensuring rapid convergence to the desired state or trajectory, even in the presence of time-varying additive bounded external disturbances. This property of exponential stabilization guarantees a quicker and more robust response when compared to traditional PID controllers.
\end{remark}

\section{Results\label{sec:results}}
In this section, we verify the efficacy of our proposed tracking controller \eqref{eqn:feedback_control_law} on five realistic scenarios.

\subsection{Mars Cube One (MarCO) mission}
NASA's Mars Cube One (MarCO) mission, launched in 2018, included two CubeSats—MarCO-A and MarCO-B—that showcased significant advancements in satellite technology. The two CubeSats possessed 6-DOF and had the ability to control both translational and rotational motion in deep space. We consider a typical Mars Cube One (MarCO) CubeSat with mass $m=13.5$kg.
The moment of inertia of MarCO CubeSat is as follows :
\begin{align}
    \bar{I}^\mathrm{B}=\begin{bmatrix}
    0.0465&-0.0007&0.0004\\
    -0.0007&0.0486&-0.0021\\
    0.0004&-0.0021&0.0482
    \end{bmatrix}\text{kg.m}^2\nonumber
\end{align}
\subsubsection{Trajectory tracking}
The objective is for the MarCO to track the desired trajectory governed by $\qhdi=1$ and $\ohdi=0$. The control gains for the feedback controller \eqref{eqn:feedback_control_law} are selected to be $k_p=0.2$ and $k_d=0.3$. 100 trajectories are generated from 100 initial conditions that are uniformly sampled inside the ball $\mathcal{B}_R$ (defined by \eqref{eqn:ball_defn}) with $R=2.5$. Figures \ref{fig:qe} and \ref{fig:omegae} show the plots of error quaternions $\|\qhbd-1\|$ and $\|\ohbd\|$ versus time respectively for these 100 trajectories. The dotted lines in Figs. \ref{fig:qe}, \ref{fig:omegae}, and \ref{fig:error_sges} show the decaying exponential upper bound. It can be observed that the errors $\|\qhbdmone\|$ and $\|\ohbd\|$ lie below these exponential curves (where the rate of convergence is governed by $\alpha(R)$) which verifies our claim that the feedback controller guarantees semi-global exponential stability.
\begin{figure}[H]
 \captionsetup[subfigure]{justification=centering}
 \centering
 \begin{subfigure}{0.33\textwidth}
{\includegraphics[scale=0.3]{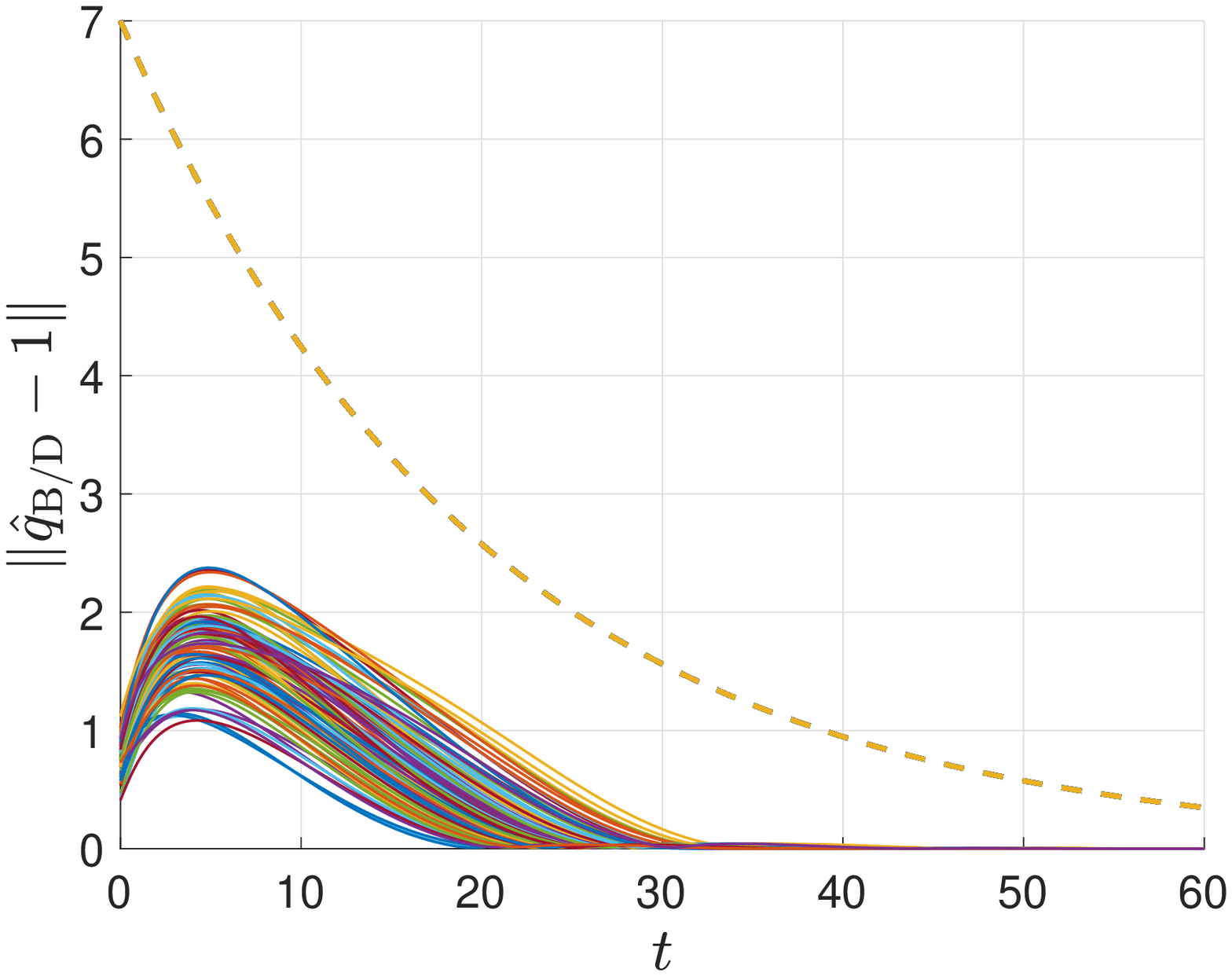}}
 \caption{$\|\qhbd-\mathbf{1}\|$ v/s $t$}
\label{fig:qe}
 \end{subfigure}
 \begin{subfigure}{0.33\textwidth}
{\includegraphics[scale=0.3]{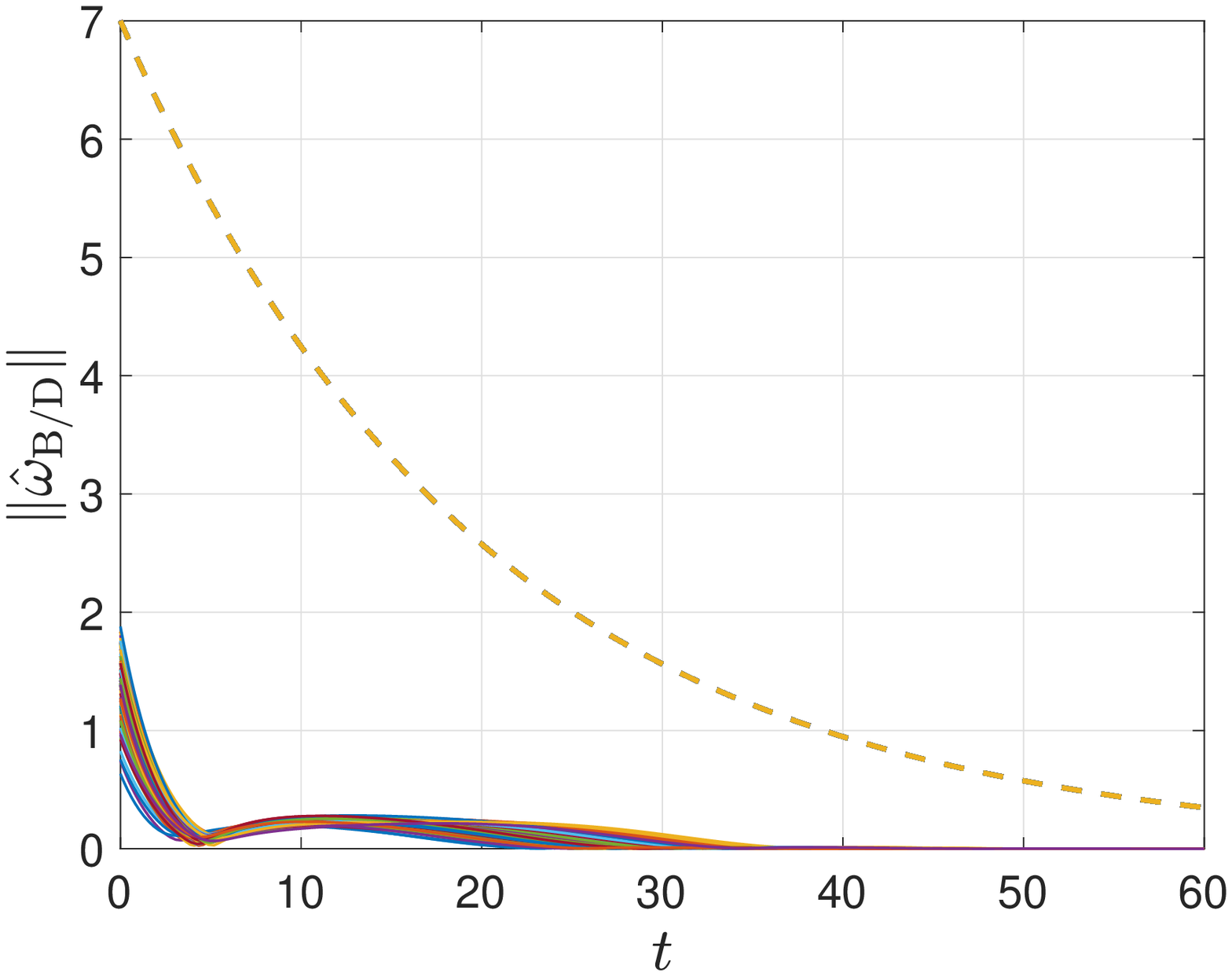}}
\caption{$\|\ohbd\|$ v/s $t$}
\label{fig:omegae}
\end{subfigure}
 \begin{subfigure}{0.33\textwidth}
{\includegraphics[scale=0.3]{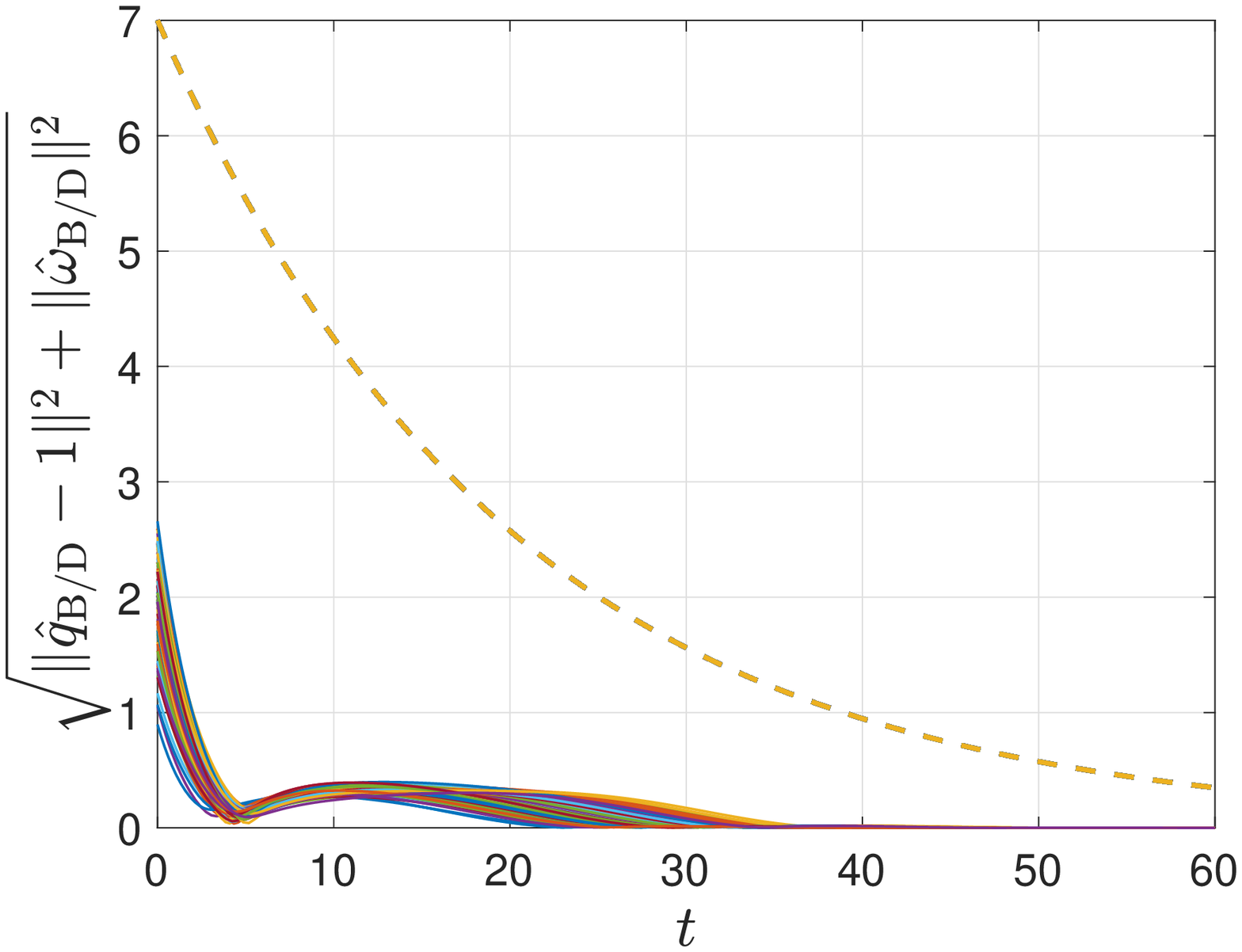}}
\caption{$\sqrt{\|\ohbd\|^2+\|\qhbd-\mathbf{1}\|^2}$ v/s $t$}
\label{fig:error_sges}
\end{subfigure}
\caption{}
\label{fig:exp_character}
\end{figure}
\subsubsection{Robustness against disturbances}
We evaluate the performance of our proposed feedback control law \eqref{eqn:feedback_control_law} for the system with time-varying additive bounded disturbances \eqref{eqn:system_dynamics_disturbances}. 100 trajectories are simulated with initial conditions chosen from a ball of radius $3.5m$, The values for $d_1(t)$ and $d_2(t)$ are chosen uniformly from the set $\mathcal{D}=[-10^{-2},\;10^{-2}]$. Note that this is a reasonable assumption to make as the external disturbances for practical applications are of the order of $10^{-2}$ to $10^{-3}$ as reported in \cite{roscoe2018overview_cubesats}.  As shown in Figs. \ref{fig:error_sges_disturbance} and \ref{fig:error_sges_disturbance_ball}, all these trajectories converge exponentially to a ball of radius less than or equal to $\psi d_m$ where $\psi$ is defined by \eqref{eqn:psi_expression}. This verifies the ISS claim made in Theorem \ref{thm:local_iss}.

\begin{figure}[H]
 \captionsetup[subfigure]{justification=centering}
 \centering
 \begin{subfigure}{0.5\textwidth}
{\includegraphics[scale=0.46]{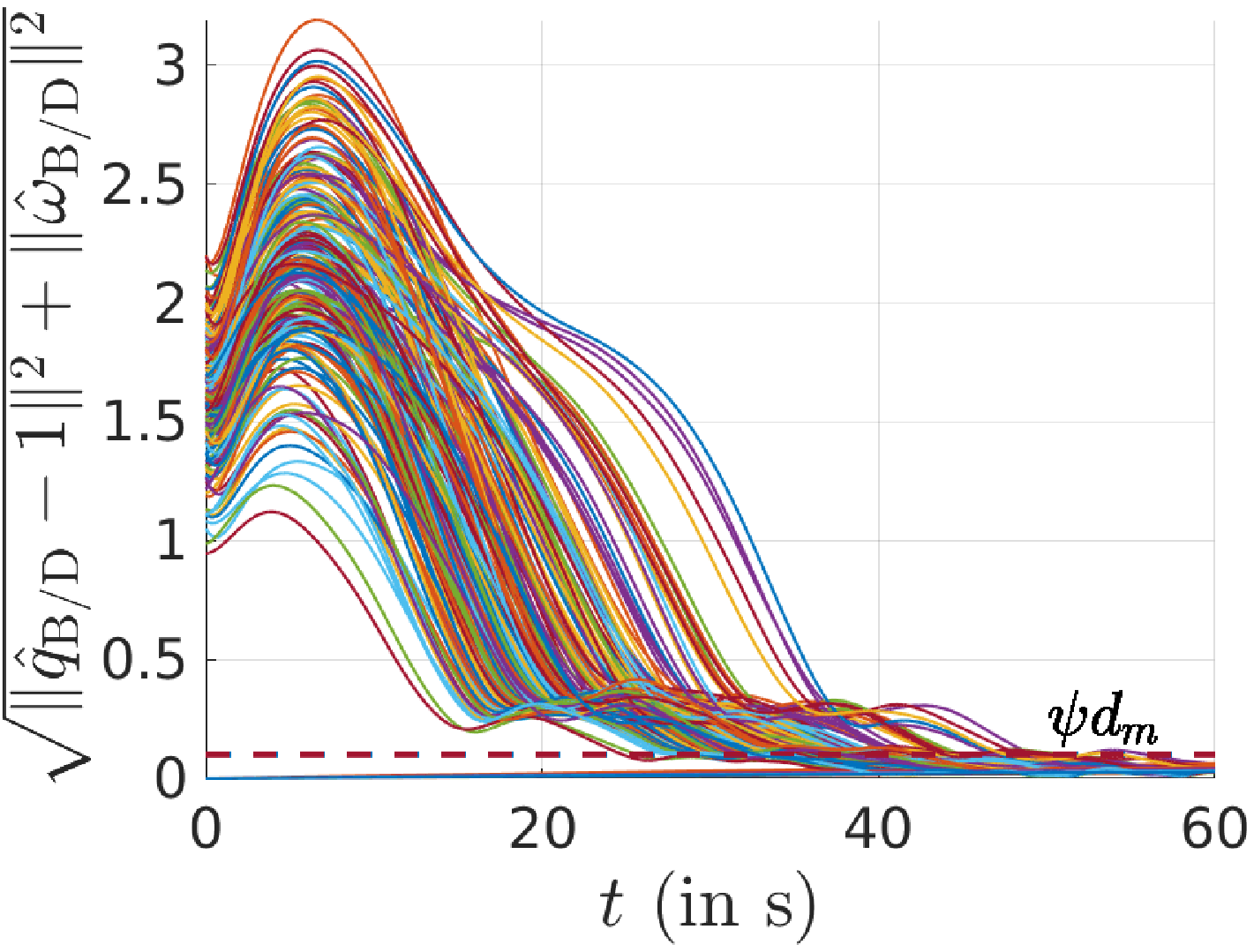}}
 \caption{$\sqrt{\|\ohbd\|^2+\|\qhbd-\mathbf{1}\|^2}$ v/s $t$}
\label{fig:error_sges_disturbance}
 \end{subfigure}
 \begin{subfigure}{0.45\textwidth}
{\includegraphics[scale=0.36]{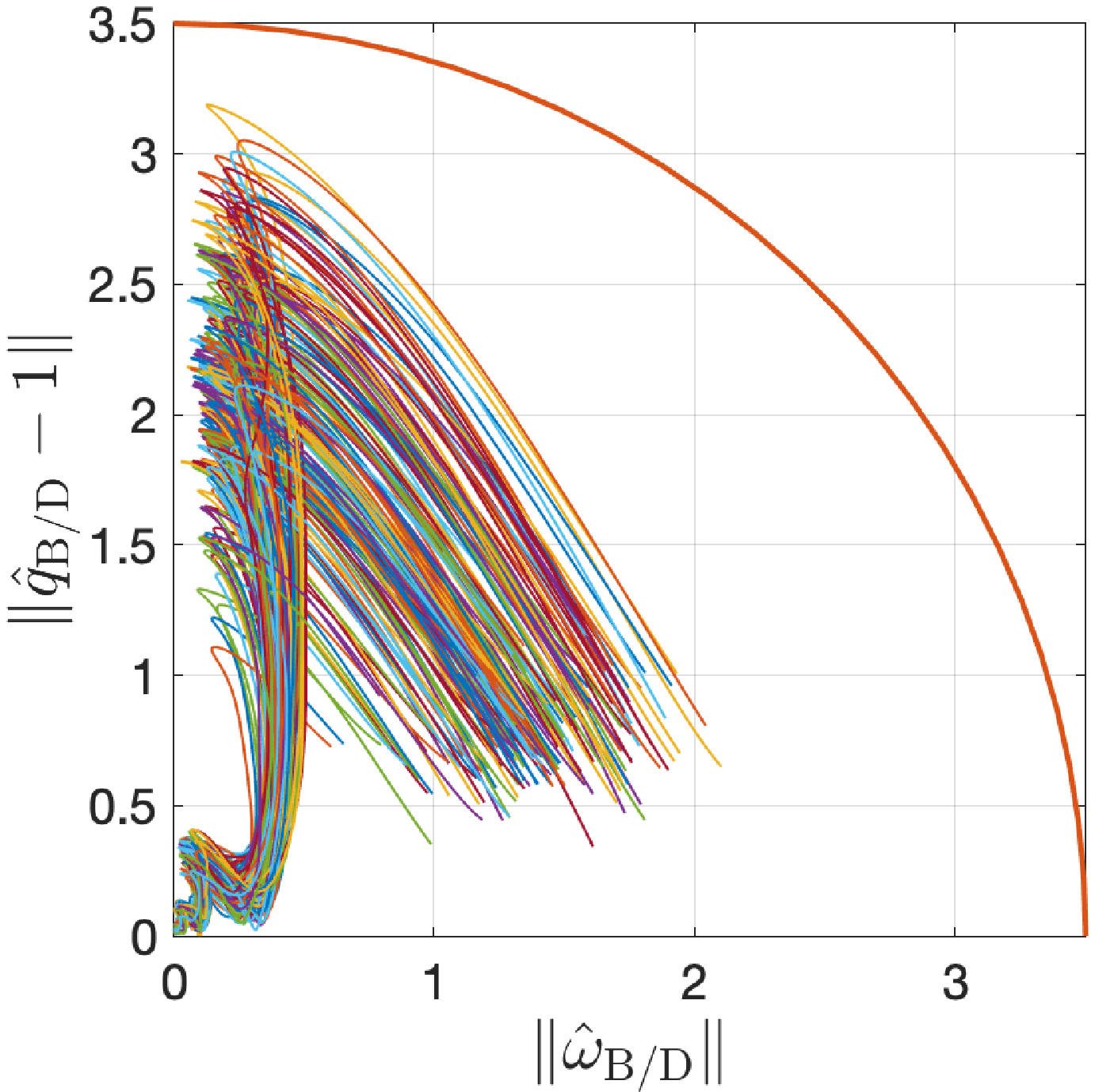}}
\caption{$\|\qhbd-\mathbf{1}\|$ vs $\|\ohbd\|$}
\label{fig:error_sges_disturbance_ball}
\end{subfigure}
\caption{}
\label{fig:exp_character_disturbances}
\end{figure}
\subsection{Apollo Transposition and Docking}
 The Apollo Transposition and Docking maneuver was a crucial step during the Apollo missions to the Moon, carried out by NASA between 1968 and 1972. It involved separating (Fig. \ref{fig:real_csm_seperation}) the Apollo Command/Service Module (CSM) from the Saturn V rocket's third stage (S-IVB), rotating the CSM, and then docking with the Lunar Module (LM), which was housed within the S-IVB.
The transposition of the CSM occurs at a distance of approximately $20m$ from the S-IVB (Fig. \ref{fig:transposition_real}). The CSM then rotates by an angle of $180^\circ$. Then the CSM docks (Fig. \ref{fig:real_docking}) in with the S-IVB to extract the lunar module (LM) and forms the combined CSM-LM spacecraft (Fig. \ref{fig:real_LM_extraction})\footnote{Image edited from \href{https://commons.wikimedia.org}{https://commons.wikimedia.org}}. 
\begin{figure}[H]
 \captionsetup[subfigure]{justification=centering}
 \centering
 \begin{subfigure}{0.49\textwidth}
{\includegraphics[scale=0.085]{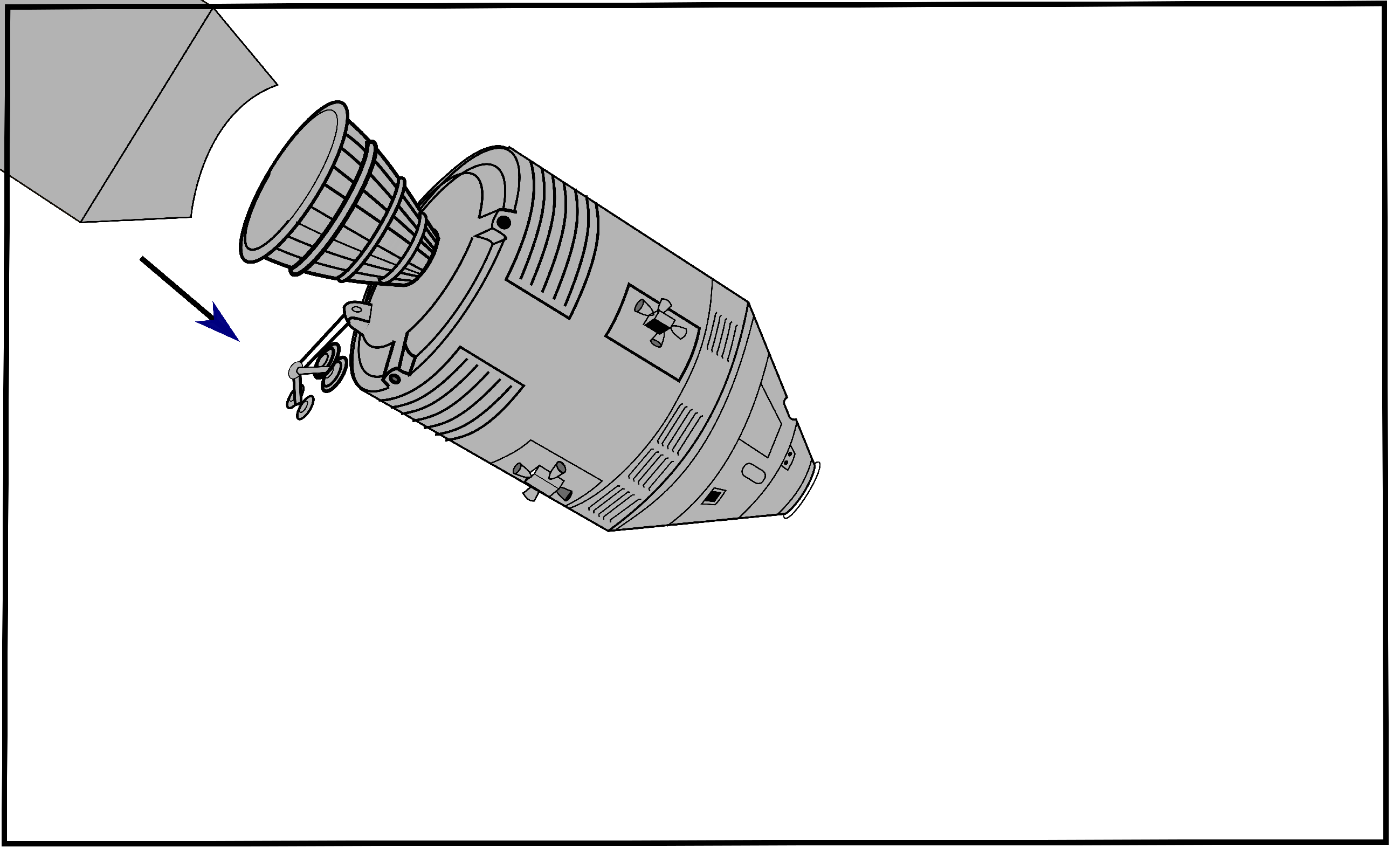}}
 \caption{Spacecraft separation}
\label{fig:real_csm_seperation}
 \end{subfigure}
 \begin{subfigure}{0.43\textwidth}
{\includegraphics[scale=0.085]{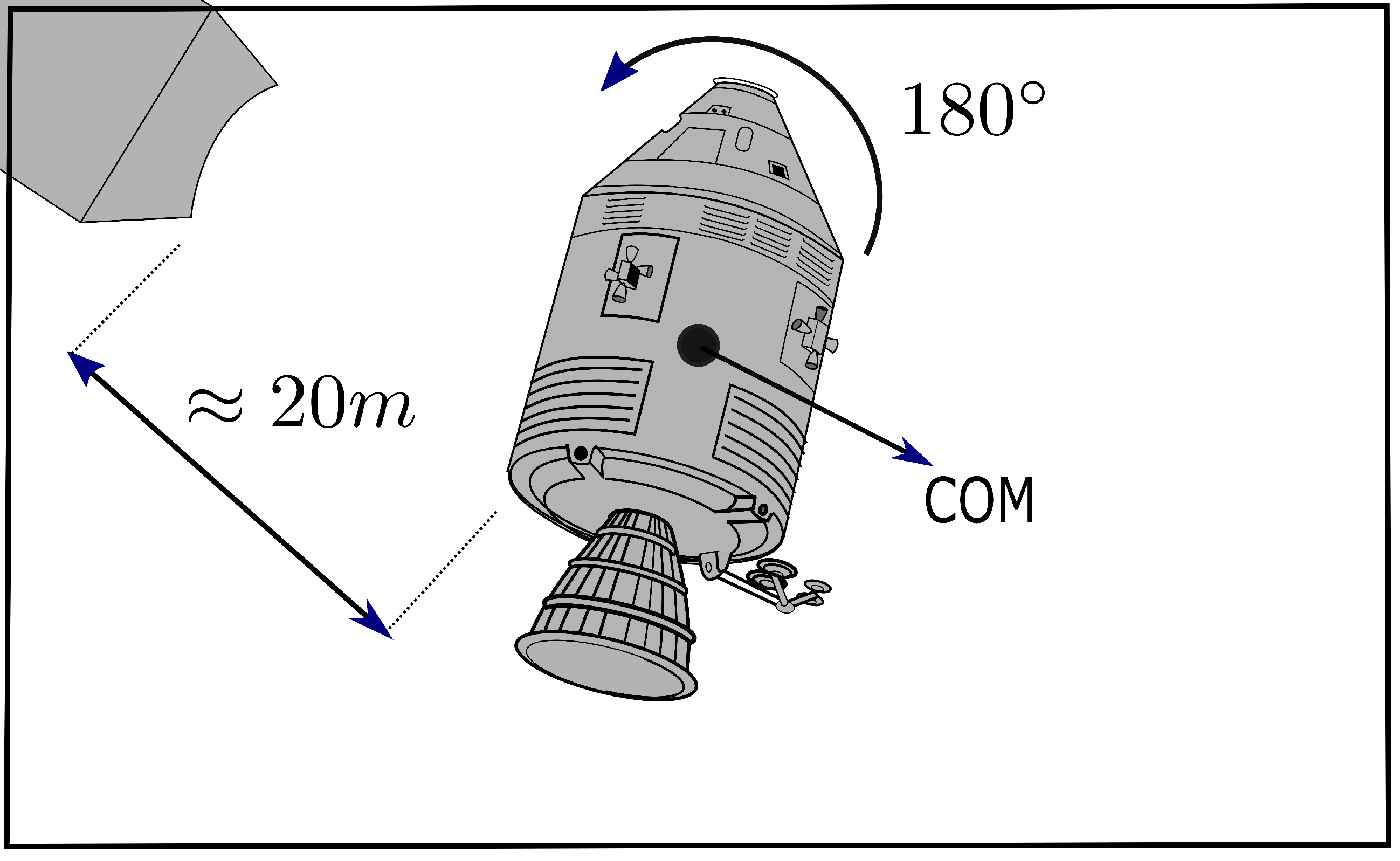}}
\caption{Transposition}
\label{fig:transposition_real}
\end{subfigure}
 \begin{subfigure}{0.49\textwidth}
{\includegraphics[scale=0.085]{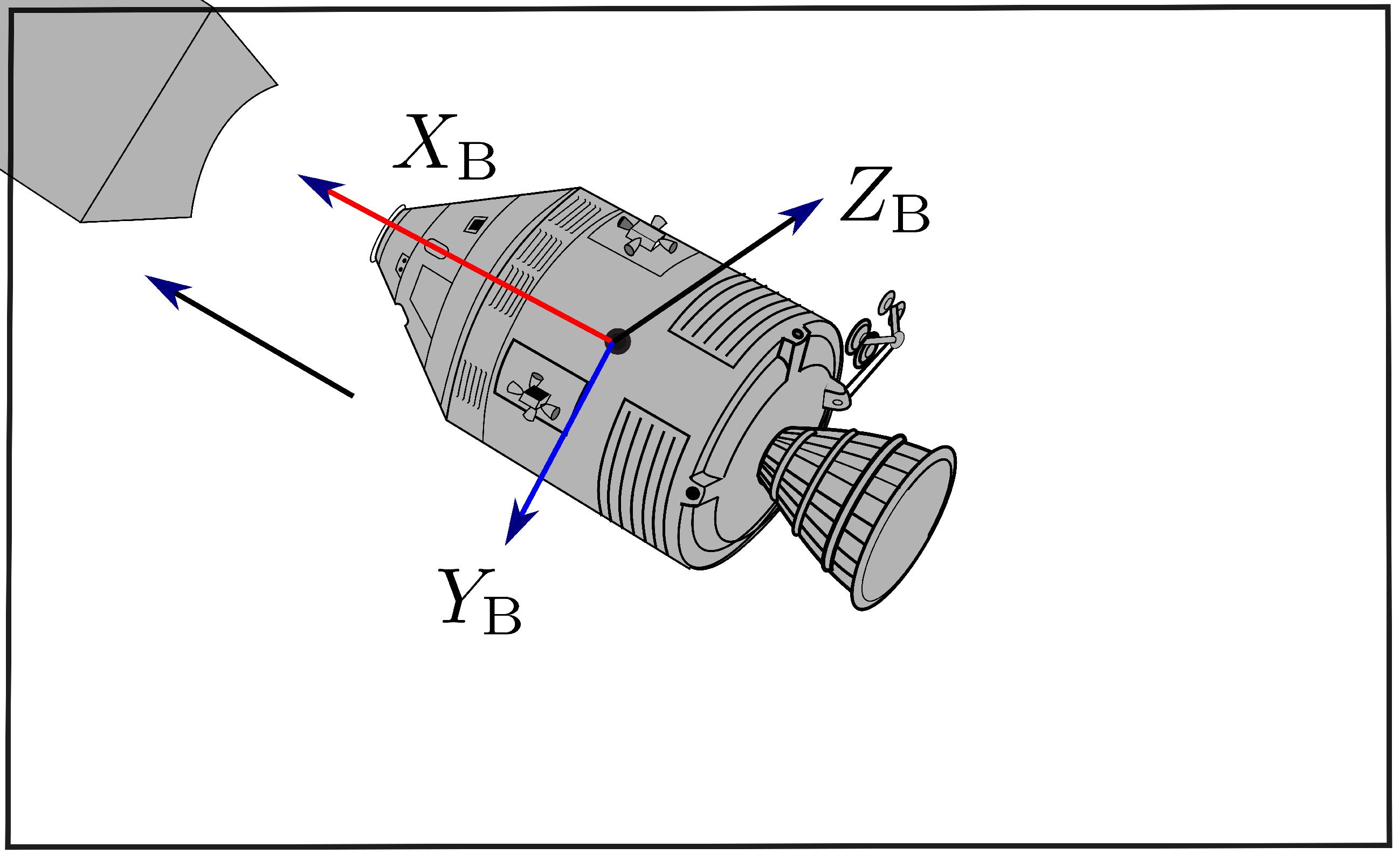}}
\caption{Docking}
\label{fig:real_docking}
\end{subfigure}
 \begin{subfigure}{0.43\textwidth}
{\includegraphics[scale=0.085]{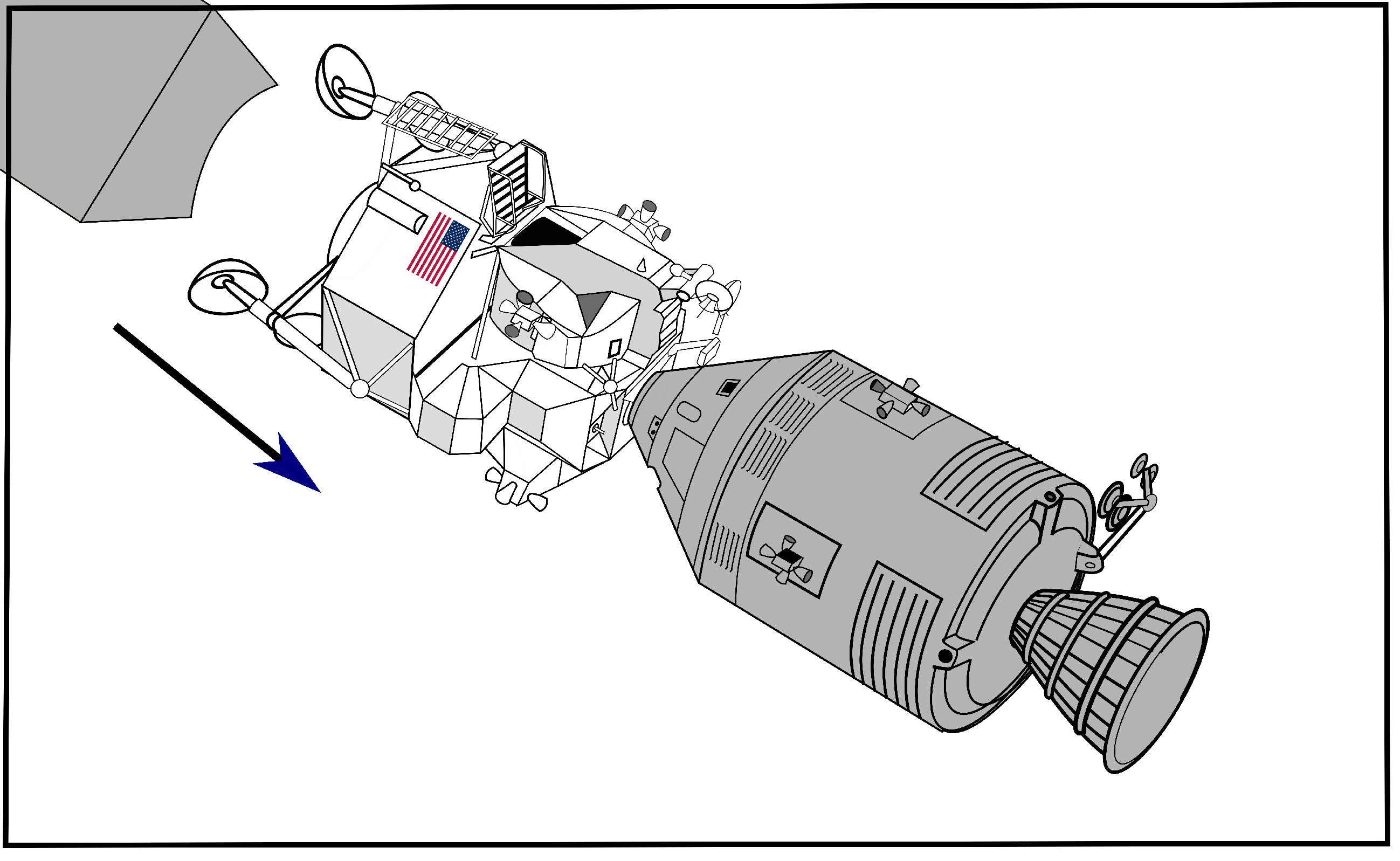}}
\caption{Lunar Module (LM) extraction}
\label{fig:real_LM_extraction}
\end{subfigure}
 \caption{Outline of the steps involved in the Apollo transposition and docking maneuver (a) Spacecraft separation: Post-TLI burn, the CSM departs from the S-IVB stage. (b) Transposition: CSM executes a $180^\circ$ pitch to align with the LM on the S-IVB stage. (c) Docking: Astronauts maneuver the CSM to dock with the LM using RCS thrusters, docking probes, and drogue. (d) Lunar Module extraction: After successful docking, the LM is detached from the S-IVB stage, preparing the Apollo stack for the lunar journey.
 }
\label{fig:steps_apollo_transposition_and_docking}
\end{figure}

We apply our proposed feedback control law \eqref{eqn:feedback_control_law} to perform the Apollo Transposition and Docking maneuver. The gains $k_p$ and $k_d$ are chosen to be equal to $10^2$. The CSM mass and inertia at the start of transposition and docking are specified in [\cite{book1969csm}, Table 3.1-2] and are given as follows:
\begin{align}
m \approx 30322.9 \mathrm{~kg}, \quad \Bar{I}^\mathrm{B} \approx\left[\begin{array}{ccc}
49248.7 & 2862.1 & -370.1 \\
2862.1 & 108514.2 & -3075.0 \\
-370.1 & -3075.0 & 110771.7
\end{array}\right] \mathrm{kg} \cdot \mathrm{m}^2\nonumber
    \end{align}
    \subsubsection{Apollo Transposition and Docking maneuver}
For the transposition maneuver, Figs. \ref{fig:apollo_transp_quaternion} and \ref{fig:apollo_torque_transp} shows the evolution of quaternion with time for a $180^\circ$ rotation and the corresponding torque $\tau^\mathrm{B}$ applied to the CSM respectively. Fig. \ref{fig:apollo_transp_exp_conv} verifies the SGES claim made by Theorem \ref{thm:es}. For the docking maneuver, 50 trajectories are generated for initial conditions sampled within a ball of radius 0.1 centered at the origin. Fig. \ref{fig:docking_position_plot} shows the evolution of the position of CSM along the docking axis with time. As observed, the terminal position of CSM is within the permissible range i.e. $p_e=0.3m$ [\cite{book1970csm}, Section 3.8.2.3]. Fig. \ref{fig:docking_input_force_plot} shows the corresponding input force applied along the docking axis.
\begin{figure}[H]
 \captionsetup[subfigure]{justification=centering}
 \centering
 \begin{subfigure}{0.33\textwidth}
{\includegraphics[scale=0.33]{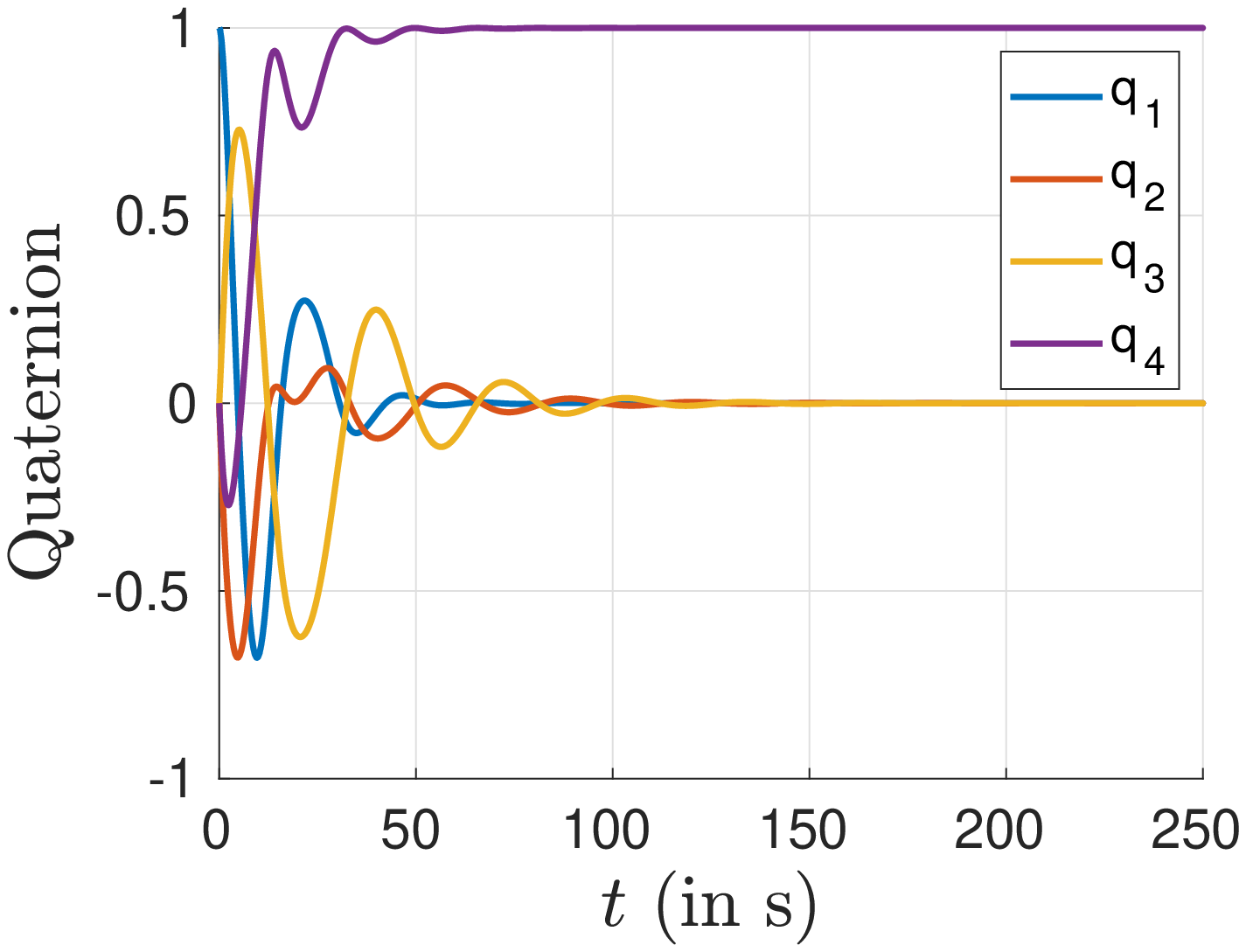}}
 \caption{Evolution of quaternion with $t$.}
\label{fig:apollo_transp_quaternion}
 \end{subfigure}
 \begin{subfigure}{0.32\textwidth}
{\includegraphics[scale=0.33]{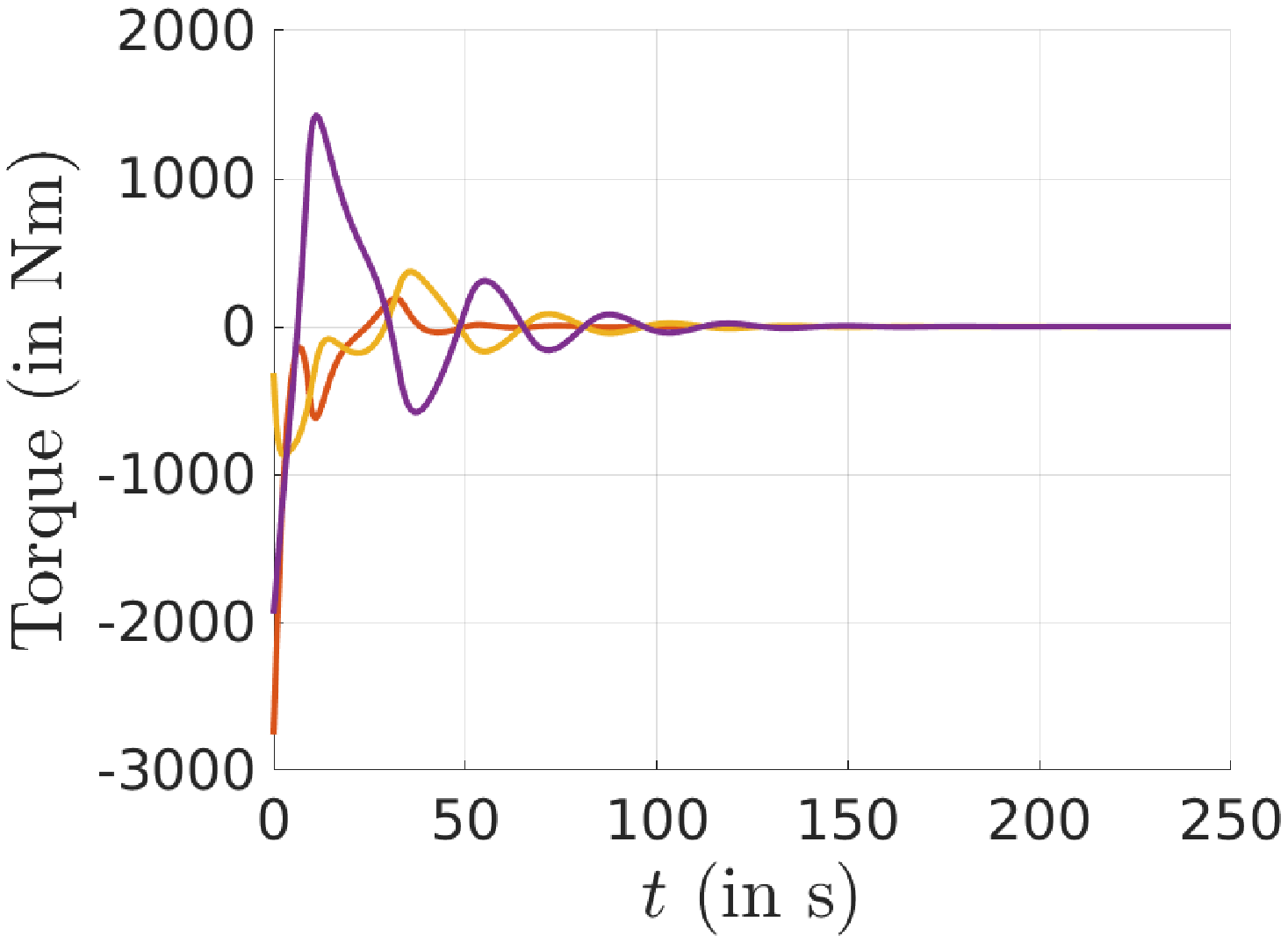}}
\caption{Evolution of torque with $t$}
\label{fig:apollo_torque_transp}
\end{subfigure}
 \begin{subfigure}{0.33\textwidth}
{\includegraphics[scale=0.33]{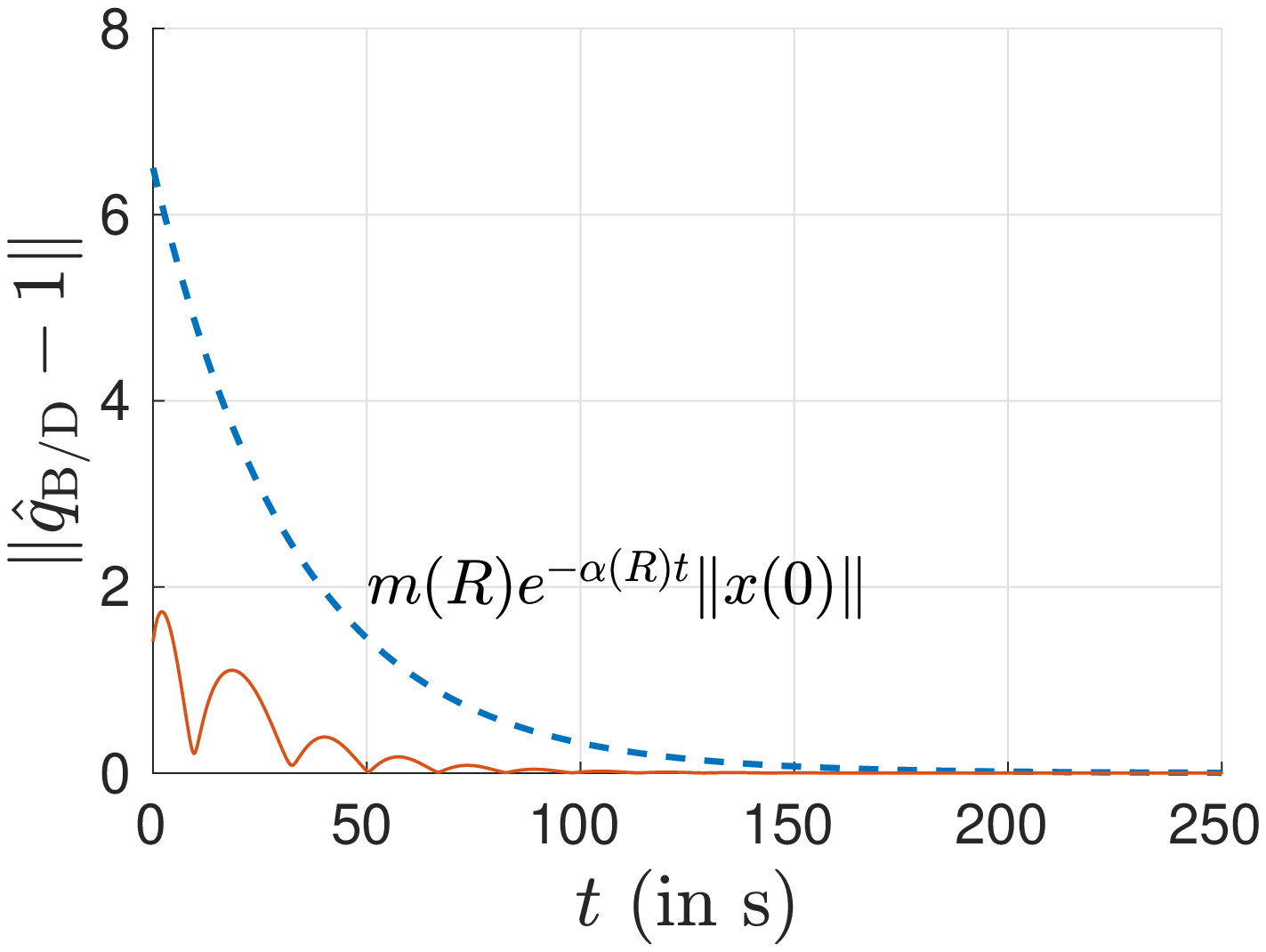}}
\caption{$\qe$ along $x$ versus $t$}
\label{fig:apollo_transp_exp_conv}
\end{subfigure}
  \caption{Evolution of quaternion $\qbd$ and torque $\tau^\mathrm{B}$ applied with time $t$ for the transposition maneuver (Fig. \ref{fig:transposition_real}). }
\end{figure}
\begin{figure}[H]
 \captionsetup[subfigure]{justification=centering}
 \centering
 \begin{subfigure}{0.45\textwidth}
{\includegraphics[scale=0.4]{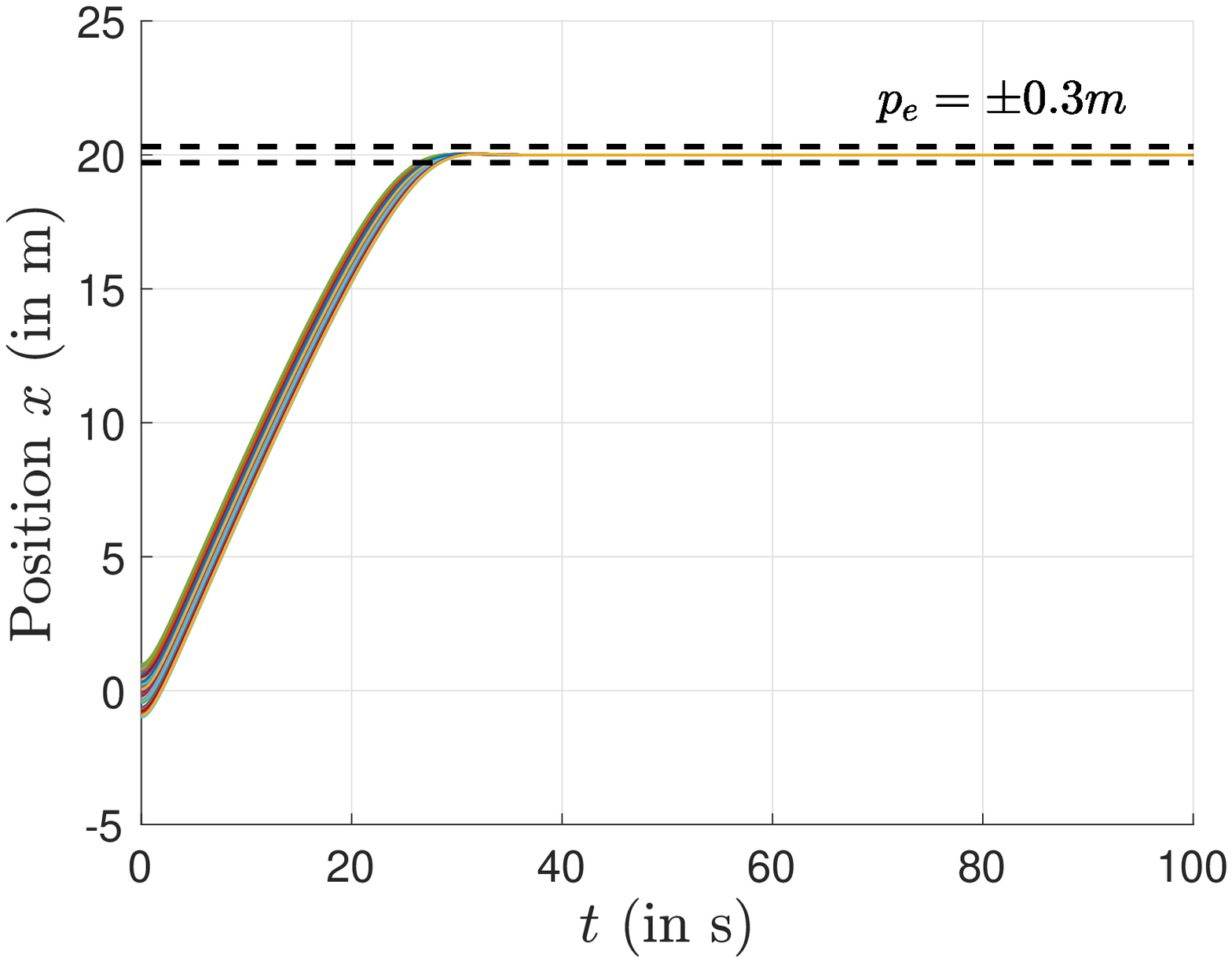}}
 \caption{Evolution of position along $x$ versus $t$.}
 \label{fig:docking_position_plot}
 \end{subfigure}
 \begin{subfigure}{0.45\textwidth}
{\includegraphics[scale=0.4]{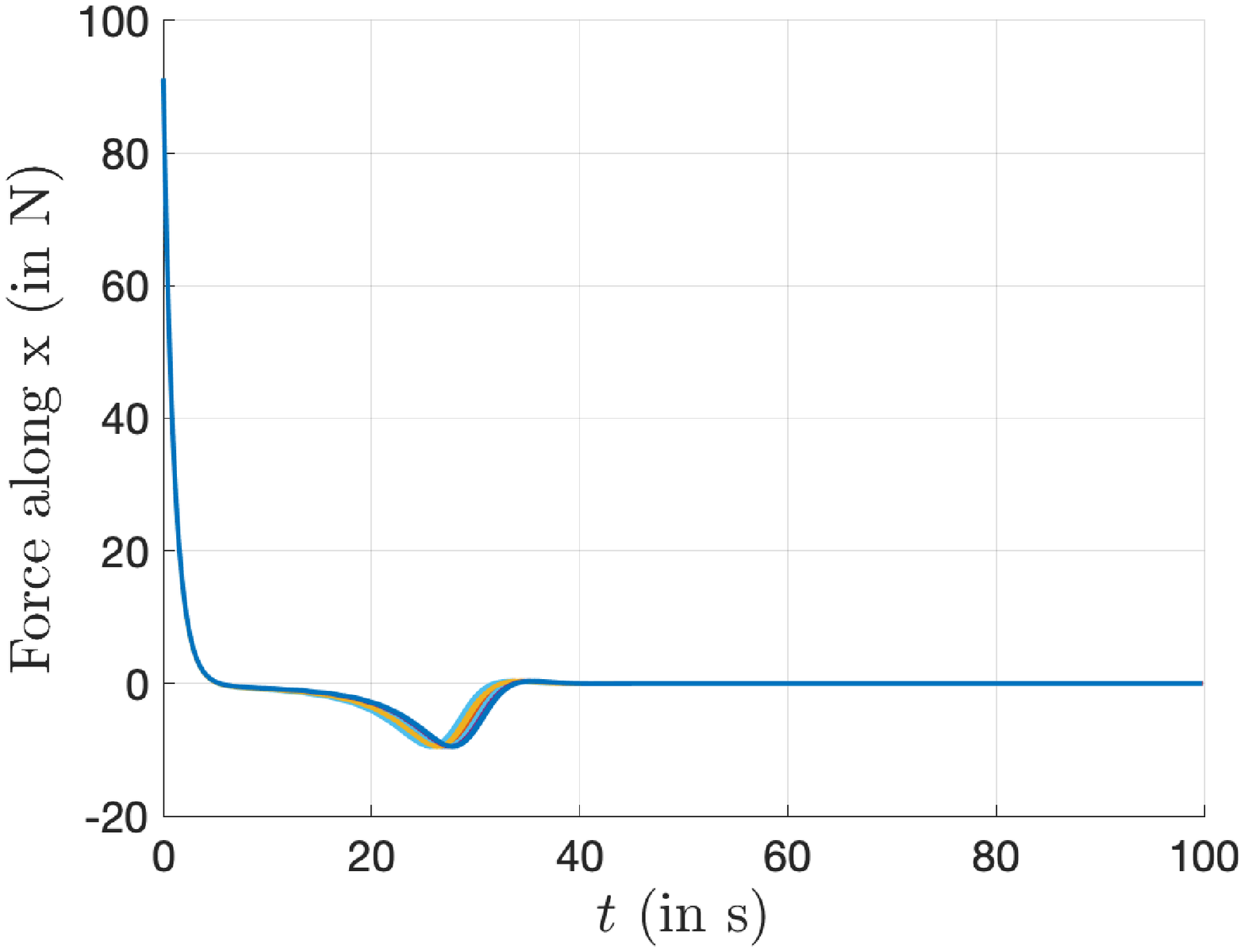}}
\caption{Input force $f^\mathrm{B}$ along $x$ versus $t$}
\label{fig:docking_input_force_plot}
\end{subfigure}
  \caption{Evolution of position and the input force $f^\mathrm{B}$ along $x$ versus $t$ for the docking maneuver (Fig. \ref{fig:real_docking}). As observed from Fig. \ref{fig:docking_position_plot}, the terminal position $x$ satisfies the permissible error of $p_e=0.3m$.}
\end{figure}

\subsubsection{Fuel consumption}
In this section, we compare the performance of our proposed controller in terms of fuel consumption with the actual fuel consumed during the Apollo transposition and docking maneuvers. The total fuel consumption of the CSM is given by the following expression:
\begin{align}
m(t)=m(0)+\int_0^tc_1\|\hat{{f}}^\mathrm{B}\|dt
\end{align}
where $c_1=1.3\times10^{-3}$kg/Ns. We apply the feedback control law \eqref{eqn:feedback_control_law} and compute the fuel consumed by CSM for final times $t=\{100,\;200,\;300,\;400\}s$. Fig. \ref{fig:fuel_consumption} shows the fuel consumption using \eqref{eqn:feedback_control_law} and the actual fuel consumed during the Apollo transposition and docking maneuver for final time $t_f=\{100,\;200,\;300,\;400\}s$. It can be clearly observed that the fuel consumed using our method is approximately $62\%$ less than that consumed during the actual Apollo mission.

\begin{figure}[H]
\centering
\includegraphics[width=11cm]{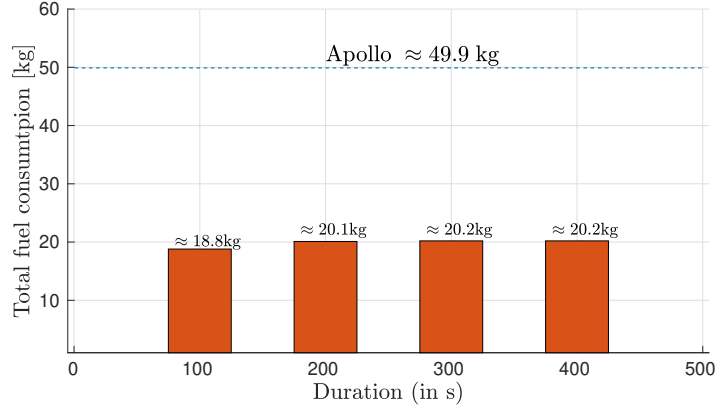}
\caption{Total fuel consumption of the trajectories generated by using the proposed feedback tracking control law \eqref{eqn:feedback_control_law}.}
\label{fig:fuel_consumption}
\end{figure}
{

\subsection{Starship flip maneuver}
The Starship (developed by SpaceX) flip maneuver advances the controlled rocket descent and landing methodology. This maneuver is executed during the terminal phase of landing, following atmospheric re-entry where the vehicle adopts a belly-first orientation to maximize air resistance and thermal shielding. As the Starship approaches the landing site, it performs an aerodynamic inversion or flip to transition from a horizontal descent to a vertical orientation, a maneuver demanding precise control over the vehicle's aerodynamic surfaces. Subsequently, the rocket engines are reignited to decelerate the spacecraft for a soft vertical landing. The initial condition for the state and the parameters of the starship are taken from \cite{malyuta2022convex_tutorial}. The comparative analysis of trajectories, as depicted in Fig. \ref{fig:starship} illustrates the distinct impact of control strategies on the system's performance. While the open-loop control, as shown in \ref{fig:starship_open}, demonstrates a tendency to violate critical constraints like glide slope under the influence of unknown smooth additive bounded disturbances, the integration of the proposed feedback controller with CBF, as seen in \ref{fig:starship_closed}, significantly mitigates these deviations, maintaining system safety/constraint satisfaction even in the presence of such disturbances. Further, the trajectory in \ref{fig:starship_closed} converges to a region given in Theorem \ref{thm:local_iss}.
\begin{figure}[H]
 \captionsetup[subfigure]{justification=centering}
 \centering
 \begin{subfigure}{0.49\textwidth}
{\includegraphics[scale=0.4]{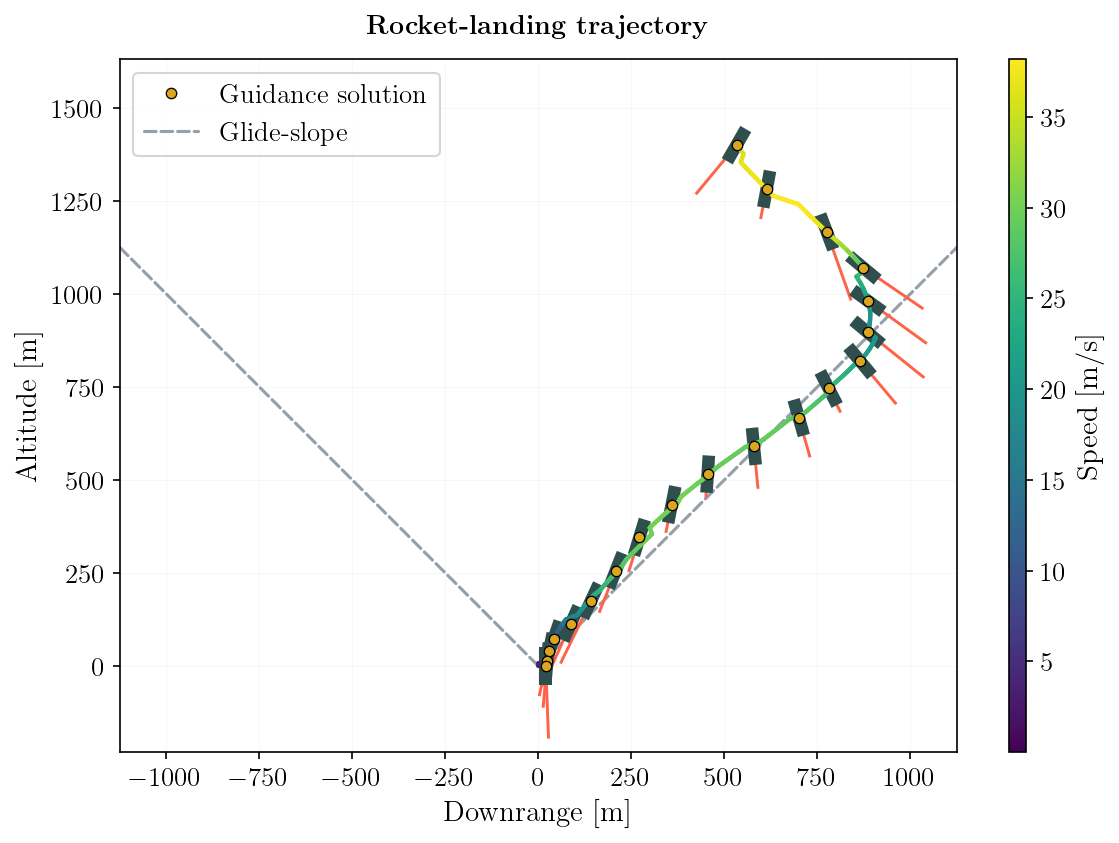}}
 \caption{Trajectory  using open loop control}
\label{fig:starship_open}
 \end{subfigure}
 \begin{subfigure}{0.49\textwidth}
{\includegraphics[scale=0.4]{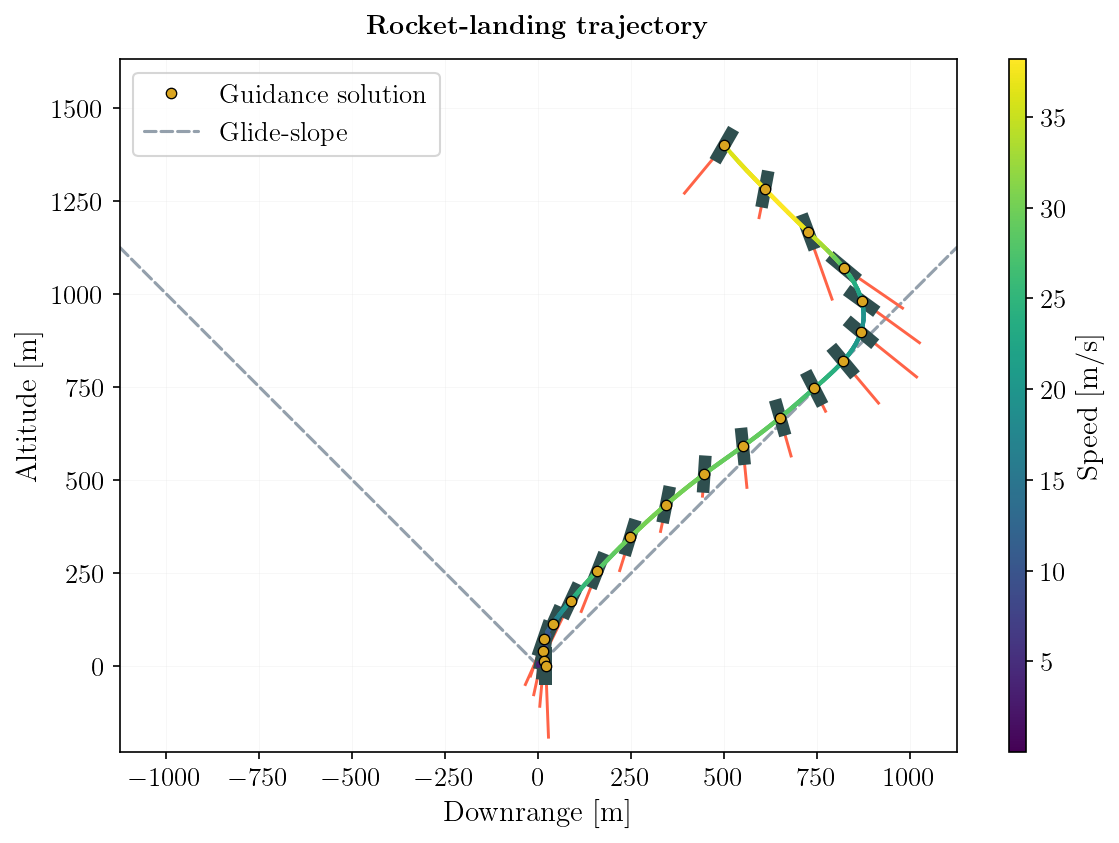}}
\caption{Trajectory using proposed feedback controller and CBF}
\label{fig:starship_closed}
\end{subfigure}

\caption{(a): The open-loop control inputs generated by solving the SCP optimization problems \cite{malyuta2022convex_tutorial} can lead to violation of constraints such as glide slope when smooth unknown additive bounded disturbances act on the actual system. (b): The figure shows the evolution of the states when a feedback control obtained by integrating the proposed feedback controller and CBF is applied to the actual system in the presence of smooth unknown additive bounded disturbances.  }
\label{fig:starship}
\end{figure}
}

{
\subsection{Spherical rocks avoiding collision avoidance with martian rock}
In the simulation setup to evaluate collision avoidance for spherical robots with a Martian rock, a realistic Martian Rochette rock and the 6-DOF spherical robot dynamics from \cite{sabet2020dynamic_spherical_robot} were used. We leverage Control Barrier Functions (CBFs) and the proposed feedback controller to navigate this environment, avoid collisions, and reach the goal. This scenario was designed to test the efficacy of the integrated feedback and CBF bases controller to avoid an obstacle (martian rock in this case) providing an assessment of their real-world application. In Fig. \ref{fig:rock_side}, we use the proposed controller \eqref{eqn:feedback_control_law} and two CBF's $h_1(\bx)={z}$ and $h_2(\bx)=H_m-z$, ($h_1(\bx)>0$ and $h_2(\bx)$ implies safety/collision avoidance) in \eqref{eqn:cbf_qp_general}x where $H_m$ is the maximum height of the rock and $z$, is the height of the spherical robot from the martian surface to avoid collision with the rock. In Fig. \ref{fig:rock_above} we only use one CBF $h_1(\bx)={z}$ and the proposed controller \eqref{eqn:feedback_control_law}. 
\begin{figure}[H]
 \captionsetup[subfigure]{justification=centering}
 \centering
 \begin{subfigure}{0.49\textwidth}
{\includegraphics[scale=0.4]{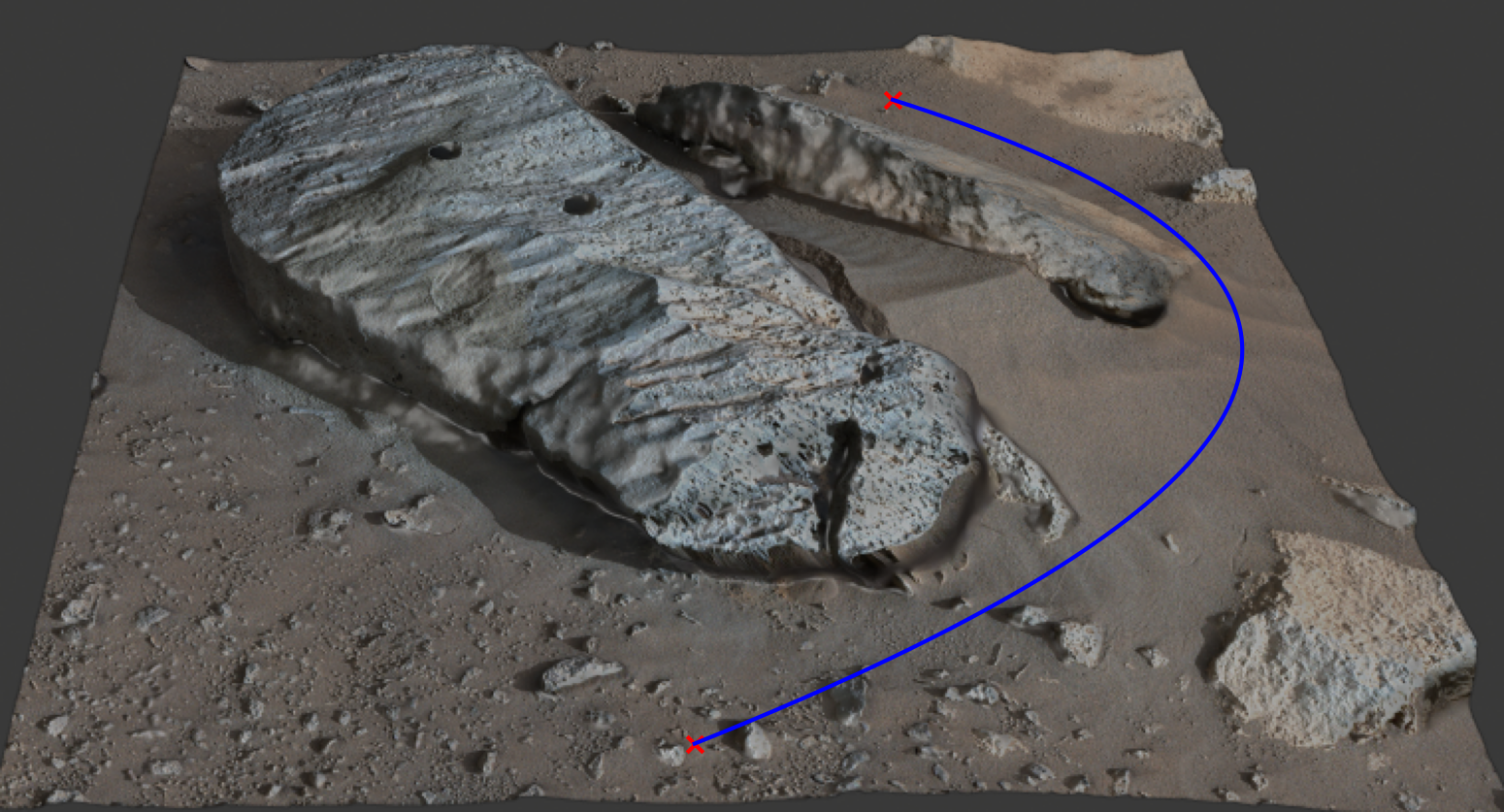}}
 \caption{}
\label{fig:rock_side}
 \end{subfigure}
 \begin{subfigure}{0.49\textwidth}
{\includegraphics[scale=0.4]{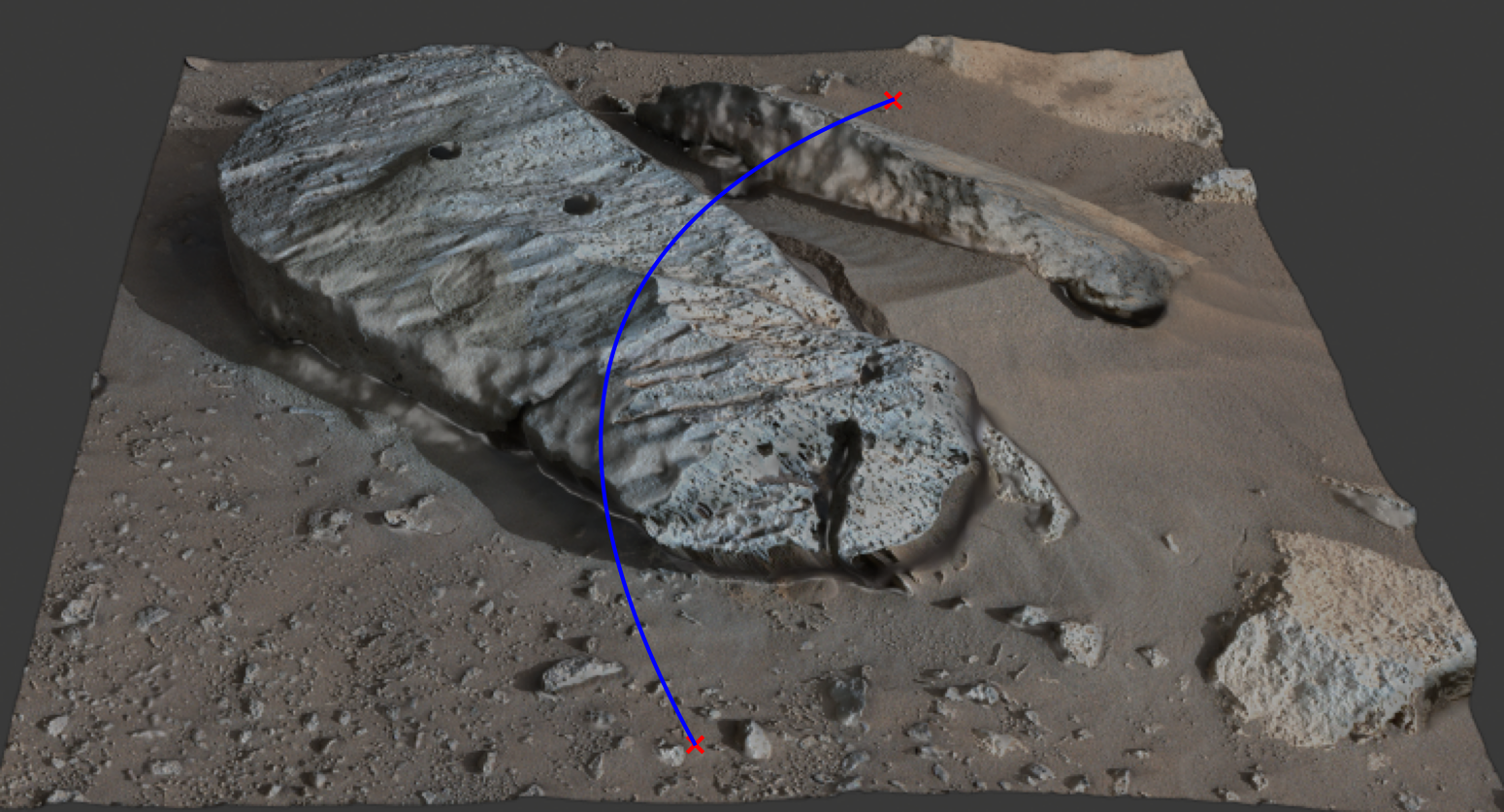}}
\caption{}
\label{fig:rock_above}
\end{subfigure}
\caption{The task for the spherical robot is to go from one point of the Martian rock to the other using the proposed integrated feedback and the CBF based controller.}
\label{fig:}
\end{figure}
}
\subsection{Docking and rendezvous of the SpaceX Dragon 2 
spacecraft with the International Space Station}
In this section, we implement our proposed nonlinear feedback tracking control law within the highly realistic, web-based simulator created by SpaceX. An illustrative screenshot of the simulator can be seen in Fig. \ref{fig:iss_simulator}\footnote{Screenshot from \href{https://iss-sim.spacex.com/}{https://iss-sim.spacex.com/}}. Our goal is to enable the Dragon 2 spacecraft to autonomously and safely dock with the International Space Station, without the need for human intervention. The following video \href{https://drive.google.com/file/d/1mQise2-45HY3LqHLdu-ovYfnYchMEG1f/view?usp=sharing}{https://drive.google.com/file/d/1mQise2-45HY3LqHLdu-ovYfnYchMEG1f/view?usp=sharing} demonstrates the successful implementation of our control law \eqref{eqn:feedback_control_law}, on Dragon 2. As seen in this video, the proposed control law allows the Dragon 2 to successfully dock with the International Space Station.
\begin{figure}[H]
\centering
\includegraphics[width=11cm]{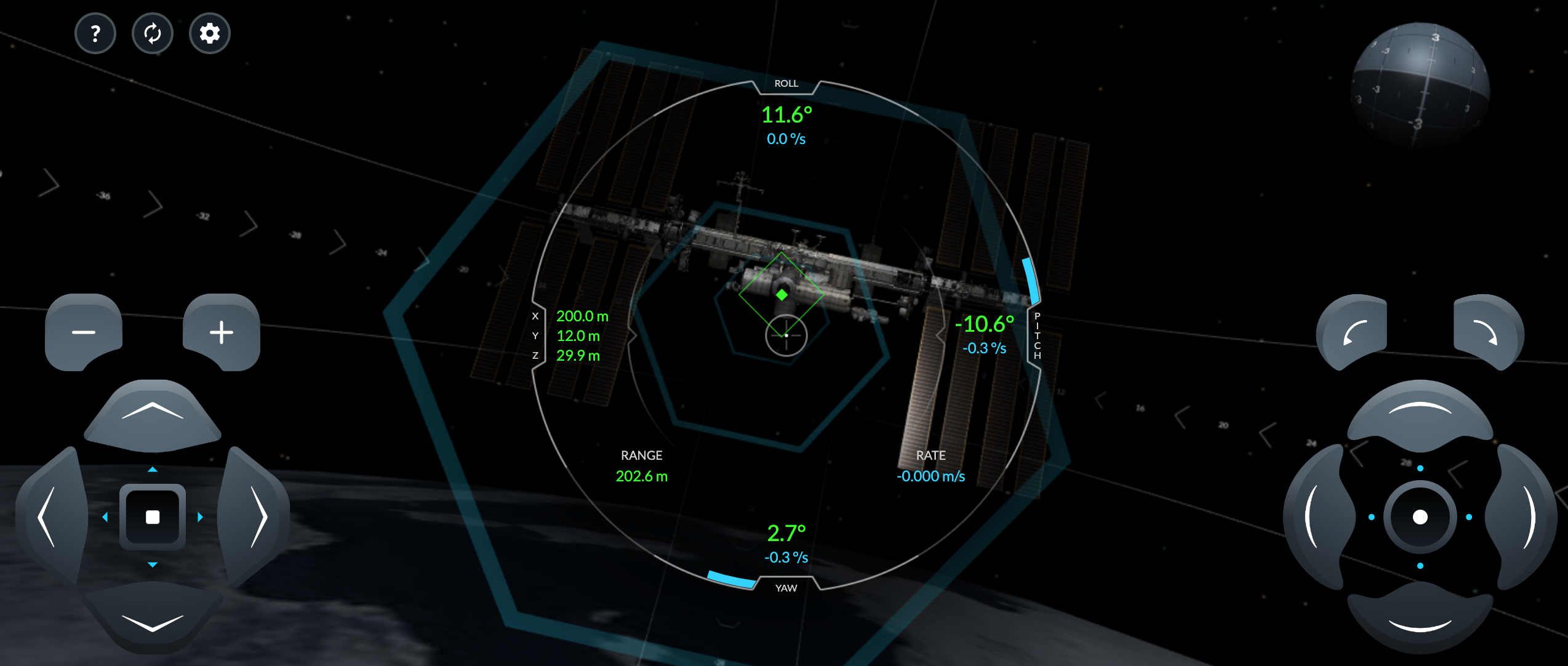}
\caption{A screenshot of the SpaceX docking simulator.}
\label{fig:iss_simulator}
\end{figure}

\section{Conclusions\label{sec:conclusions}}
A novel dual quaternion based nonlinear feedback tracking controller is proposed in this paper. Using Lyapunov analysis, Semi-Global Exponential Stability is established for the closed-loop dual quaternion based tracking error dynamics without imposing any extra constraints on the feedback gains. Further, we perform robustness analysis by proving Input-to-State Stability for the closed-loop dynamics in the presence of time-varying additive and bounded external disturbances. Using the proposed controller as the nominal input, we leverage Control Barrier Functions to compute safe control inputs that ensure that motion and safety constraints are satisfied. Finally, we verify the efficacy of the proposed tracking controller in realistic scenarios such as the MarCO mission, the Apollo transposition and docking problem, Starship flip maneuver, the collision avoidance of spherical robots, and the docking of SpaceX Dragon 2 with the International Space Station. Future work includes redesigning the proposed tracking controllers so that their applicability does not require any model information such as mass and moment of inertia.

\bibliography{main}

\end{document}